\documentclass[final,3p]{elsarticle}
\usepackage{graphicx,amssymb,lineno,url,latexsym,amsmath,url,color,amsthm}

\newcommand{\set}[1]{\left\{#1\right\}}

\newcommand{\p}{\partial}
\newcommand{\eps}{\varepsilon}
\newcommand{\vB}{\mathbf{B}}
\newcommand{\vE}{\mathbf{E}}
\newcommand{\vF}{\mathbf{F}}
\newcommand{\vH}{\mathbf{H}}

\newcommand{\vU}{\mathbf{U}}
\newcommand{\vV}{\mathbf{V}}
\newcommand{\vW}{\mathbf{W}}
\newcommand{\n}{\mathbf{n}}
\newcommand{\x}{\mathbf{x}}
\newcommand{\y}{\mathbf{y}}
\newcommand{\z}{\mathbf{z}}
\newcommand{\va}{\boldsymbol{\phi}}
\newcommand{\vv}{\boldsymbol{\vartheta}}
\newcommand{\vt}{\boldsymbol{\theta}}
\newcommand{\vtau}{\boldsymbol{\tau}}
\newcommand{\veta}{\boldsymbol{\eta}}

\theoremstyle{plain}
\newtheorem{thm}{Theorem}[section]

\newtheorem{rem}[thm]{Remark}

\begin{document}
\begin{frontmatter}
\title{Analysis of a multi-frequency electromagnetic imaging functional for thin, crack-like electromagnetic inclusions}

\author{Won-Kwang Park}
\ead{parkwk@kookmin.ac.kr}
\address{Department of Mathematics, Kookmin University, Seoul, 136-702, Korea.}

\begin{abstract}
Recently, a non-iterative multi-frequency subspace migration imaging algorithm was developed based on an asymptotic expansion formula for thin, curve-like electromagnetic inclusions and the structure of singular vectors in the Multi-Static Response (MSR) matrix. The present study examines the structure of subspace migration imaging functional and proposes an improved imaging functional weighted by the frequency. We identify the relationship between the imaging functional and Bessel functions of integer order of the first kind. Numerical examples for single and multiple inclusions show that the presented algorithm not only retains the advantages of the traditional imaging functional but also improves the imaging performance.
\end{abstract}

\end{frontmatter}

\section{Introduction}
Central to this paper is the problem of inverse scattering from thin curve-like electromagnetic inclusion(s) embedded in a homogeneous domain. Generally, the main purpose of inverse scattering problem is to identify characteristics of the target of interest, e.g., shape, location, and material properties from measured scattered field data. Among them, shape reconstruction of extended electromagnetic inhomogeneities with small thickness or perfectly conducting cracks is viewed as a difficult problem because of its ill-posedness and nonlinearily but this attracted the attention of researchers because this problem plays a significant role in many fields such as physics, medical science, and non-destructive testing of materials. Consequently, various shape reconstruction algorithms have been reported. However, most of these algorithms are based on Newton-type iteration scheme, which requires addition of a regularization term, complex evaluation of Fr{\'e}chet derivatives at each iteration step, and \textit{a priori} information of unknown inclusion(s). However, shape reconstruction via the iteration method with a bad initial guess fails if the above conditions are not fulfilled.

On account of this, various non-iterative shape reconstruction algorithms have been developed, such as MUltiple SIgnal Classification (MUSIC) algorithm \cite{AKLP,JP,PL1,PL3}, topological derivative strategy \cite{MKP,MP,P4,P3}, and linear sampling method \cite{CHM,KR}. Recently, multi-frequency based subspace migration imaging algorithm was developed for obtaining a more accurate shape of unknown inclusions. Related articles can be found in \cite{AGKPS,G,P1,P2,PL2,PP} and references therein. However, studies have applied this algorithm heuristically and therefore certain phenomena such as appearance of unexpected ghost replicas cannot be explained. This gave the main impetus for this study to explore the structure of multi-frequency imaging algorithm.

In this paper, we carefully analyze the structure of multi-frequency subspace migration imaging functional by establishing a relationship with the Bessel functions of integer order of the first kind. This is based on the fact that measured boundary data can be represented as an asymptotic expansion formula in the existence of extended, thin inclusion(s). We proceed to explore certain properties of imaging functional and find conditions for good imaging performance. We finally propose an improved imaging functional weighted by power of applied frequencies for producing better results. We carry out a structure analysis of imaging functional, explore a condition of imaging performance, discuss certain properties, and present numerical examples to show its feasibility.

The rest of this paper is organized as follows. In Section \ref{sec:2} we introduce the two-dimensional direct scattering problem and an asymptotic expansion formula for a thin electromagnetic inclusion, and then review the multi-frequency subspace migration imaging functional presented in \cite{AGJK,AGKPS,P1,P2,PL2}. In Section \ref{sec:3}, we discover and specify the structure and properties of the existing multi-frequency subspace migration imaging functional. We then design an improved imaging functional weighted by several applied frequencies and analyze its structure in order to investigate its properties. In Section \ref{sec:4}, various numerical examples are exhibited and discussed in order to verify our theoretical results. Finally, in Section \ref{sec:5}, conclusion of this paper is presented.

Finally, we would like to mention that although the constructed shape via the proposed algorithm does not match the target shape completely, considering it as an initial guess for an iterative algorithm will be helpful for a successful reconstruction; for details, refer to \cite{ADIM,DL,PL4}.

\section{Preliminaries}\label{sec:2}
\subsection{Direct scattering problem and asymptotic expansion formula}
We briefly survey two-dimensional electromagnetic scattering from a thin, curve-like inclusion in a homogeneous domain. For this purpose, let $\Omega$ and $\Upsilon$ denote a homogeneous domain with a smooth boundary and a thin inclusion, which is characterized in the neighborhood of a simple, smooth curve $\gamma$:
\[\Upsilon=\set{\x+\rho\veta(\x):\x\in\gamma,~-h\leq\rho\leq h},\]
where $h$ specifies the thickness of $\Upsilon$. Throughout this paper, we denote $\vtau(\x)$ and $\veta(\x)$ as the unit tangential and normal, respectively to $\gamma$ at $\x$ (See Figure \ref{FigureGamma}).

\begin{figure}
\begin{center}
\includegraphics[width=0.45\textwidth,keepaspectratio=true,angle=0]{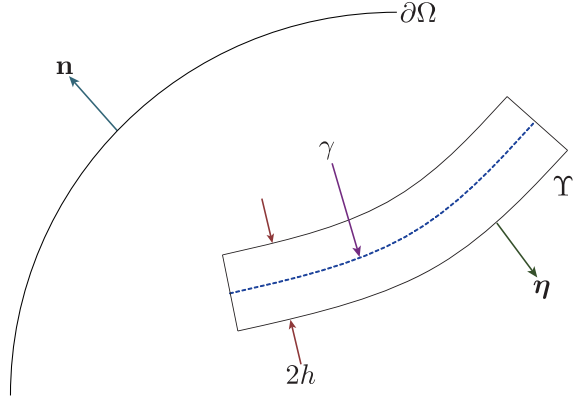}
\caption{\label{FigureGamma}Two-dimensional thin, curve-like electromagnetic inclusion $\Upsilon$ of thickness $2h$ with supporting curve $\gamma$.}
\end{center}
\end{figure}

Let $\eps_0$ and $\mu_0$ denote the dielectric permittivity and magnetic permeability of $\Omega$, respectively. Similarly, $\eps$ and $\mu$ denote the permittivity and permeability of $\Upsilon$, respectively. At a given non-zero frequency $\omega$, let $u^{(l)}(\x;\omega)$ be the time-harmonic total electromagnetic field that satisfies the following boundary value problem:
\begin{equation}\label{Forward}
  \left\{
  \begin{array}{rcl}
    \medskip\displaystyle\nabla\cdot\left(\frac{1}{\mu}\chi(\Upsilon)+\frac{1}{\mu_0}\chi(\Omega\backslash\overline{\Upsilon})\right)\nabla u^{(l)}(\x;\omega)+\omega^2\left(\eps\chi(\Upsilon)+\eps_0\chi(\Omega\backslash\overline{\Upsilon})\right)u^{(l)}(\x;\omega)=0
  & \mbox{for} & \x\in\Omega \\
  \medskip\displaystyle\frac{1}{\mu_0}\frac{\p u^{(l)}(\x;\omega)}{\p\n(\x)}=\frac{1}{\mu_0}\frac{\p e^{i\omega\vt_l\cdot\x}}{\p\n(\x)} & \mbox{for} & \x\in\p\Omega
  \end{array}
\right.
\end{equation}
with transmission conditions on $\p\Omega$ and $\p\Upsilon$. Here, $\left\{\vt_l:l=1,2,\cdots,Q\right\}$ is the set of incident directions equally distributed on the unit circle $\mathbb{S}^1$, and $\chi(A)$ denotes the characteristic function of a set $A$. Let $u_0^{(l)}(\x;\omega)=e^{i\omega\vt_l\cdot\x}$ be the solution of (\ref{Forward}) without $\Upsilon$. Then, due to the existence of $\Upsilon$, the following asymptotic expansion formula holds. This formula will contribute to development of the imaging algorithm. A rigorous derivation of this formula can be found in \cite{BF}.

\begin{thm}[Asymptotic expansion formula]
  For $\x\in\gamma$ and $\y\in\p\Omega$, the following asymptotic expansion formula holds:
  \[u^{(l)}(\y;\omega)-u_0^{(l)}(\y;\omega)=hu_\gamma^{(l)}(\y;\omega)+\mathcal{O}(h^2),\]
  where the perturbation term $u_\gamma^{(l)}(\x;\omega)$ is given by
  \begin{equation}\label{perturbation}
    u_\gamma^{(l)}(\y;\omega)=\omega^2\int_{\gamma}\left\{\left(\frac{\eps-\eps_0}{\sqrt{\eps_0\mu_0}}\right)u_0^{(l)}(\x;\omega)\Lambda(\x,\y;\omega) +\nabla u_0^{(l)}(\x;\omega)\cdot\mathbb{M}(\x)\cdot\nabla\Lambda(\x,\y;\omega)\right\}d\gamma(\x).
  \end{equation}
  Here, $\Lambda(\x,\y;\omega)$ is the Neumann function for $\Omega$ that satisfies
  \[\left\{
  \begin{array}{rcl}
  \displaystyle\frac{1}{\mu_0}\nabla\cdot\nabla \Lambda(\x,\y;\omega)+\omega^2\eps_0\Lambda(\x,\y;\omega)=-\delta(\x,\y) & \mbox{for} & \x\in\Omega \\
  \displaystyle\frac{1}{\mu_0}\frac{\p \Lambda(\x,\y;\omega)}{\p\n(\x)}=0 & \mbox{for} & \x\in\p\Omega,
  \end{array}
  \right.\]
  and $\mathbb{M}(\x)$ is a $2\times2$ symmetric matrix such that
  \begin{itemize}
    \item $\mathbb{M}(\x)$ has eigenvectors $\vtau(\x)$ and $\veta(\x)$,
    \item eigenvalues corresponding to $\vtau(\x)$ and $\veta(\x)$ are $2\left(\frac{1}{\mu}-\frac{1}{\mu_0}\right)$ and $2\left(\frac{1}{\mu_0}-\frac{\mu}{\mu_0^2}\right)$, respectively.
  \end{itemize}
\end{thm}

\subsection{Introduction to multi-frequency subspace migration imaging functional}
Now, we apply the asymptotic formula (\ref{perturbation}) in order to introduce the multi-frequency imaging functional. To do so, we will use the eigenvalue structure of the MSR matrix $\mathbb{A}(\omega)=[A_{jl}(\omega)]\in\mathbb{C}^{P\times Q}$, where the element $A_{jl}(\omega)$ is the following normalized boundary measurement:
\begin{multline}\label{ElementA}
  A_{jl}(\omega)=\int_{\p\Omega}u_\gamma^{(l)}(\y;\omega)\frac{\p v^{(j)}(\y;\omega)}{\p\n(\y)}dS(\y)
  =\omega^2\int_{\gamma}\left\{\left(\frac{\eps-\eps_0}{\sqrt{\eps_0\mu_0}}\right)+\vv_j\cdot\mathbb{M}(\x)\cdot\vt_l\right\}e^{-i\omega(\vv_j-\vt_l)\cdot\x}d\gamma(\x)\\
  \approx\frac{\omega^2\ell(\gamma)}{M}\sum_{m=1}^{M}\left\{\left(\frac{\eps-\eps_0}{\sqrt{\eps_0\mu_0}}\right)+2\left(\frac{1}{\mu}-\frac{1}{\mu_0}\right)\vv_j\cdot\vtau(\x_m)\vt_l\cdot\vtau(\x_m)\right.\\ \left.+2\left(\frac{1}{\mu_0}-\frac{\mu}{\mu_0^2}\right)\vv_j\cdot\veta(\x_m)\vt_l\cdot\veta(\x_m)\right\}e^{-i\omega(\vv_j-\vt_l)\cdot\x_m},
\end{multline}
where $\ell(\gamma)$ denotes the length of $\gamma$ and $v^{(j)}(\x;\omega)=e^{-i\omega\vv_j\cdot\x}$, $j=1,2,\cdots,P$, is a test function. Note that $\{\vv_j:j=1,2,\cdots,P\}$ is the set of unit vectors equally distributed on the two-dimensional unit circle $\mathbb{S}^1$. Based on the recent work \cite{PL3}, the number of $P$ and $Q$ must be satisfying $P,Q\gg3M$.

Throughout this section, we assume that for a given frequency $\omega=\frac{2\pi}{\lambda}$, $\gamma$ is divided into $M$ different segments of size of order $\frac{\lambda}{2}$, and only one point, say $\x_m$ for $m=1,2,\cdots,M$, at each segment, contributes to the image space of the Multi-Static Response (MSR) matrix $\mathbb{A}(\omega)$ based on the Rayleigh resolution limit; for details, refer to \cite{A,PL1,PL2,PL3}.

Based on the representation (\ref{ElementA}), $\mathbb{A}(\omega)$ can be decomposed as follows:
\begin{equation}\label{SymmetryMSR}
  \mathbb{A}(\omega)=\mathbb{B}(\omega)\mathbb{D}(\omega)\mathbb{H}(\omega),
\end{equation}
where $\mathbb{D}(\omega)\in\mathbb{R}^{3M\times3M}$ is a block diagonal matrix with components
\[\frac{\omega^2\ell(\gamma)}{M}
\left[
  \begin{array}{cc}
    \eps-\eps_0 & \mathbb{O}_{1\times2} \\
    \mathbb{O}_{2\times1} & \mathbb{M}(\x_m)\\
  \end{array}
\right],\]
and $\mathbb{B}(\omega)\in\mathbb{C}^{P\times3M}$ and $\mathbb{H}(\omega)\in\mathbb{C}^{3M\times Q}$ are written as
\[\mathbb{B}(\omega)=\bigg[\vB_1(\omega),\vB_2(\omega),\cdots,\vB_M(\omega)\bigg]\quad\mbox{and}\quad
\mathbb{H}(\omega)=\bigg[\vH_1(\omega),\vH_2(\omega),\cdots,\vH_M(\omega)\bigg],\]
respectively. Here, $\mathbb{O}_{p\times q}$ denotes the $p\times q$ zero matrix and vectors $\vB_m(\omega)$ and $\vH_m(\omega)$ are represented
\begin{equation}\label{Vec1}
  \vB_m(\omega)=\left[
        \begin{array}{ccc}
          e^{-i\omega\vv_1\cdot\x_m}, & \vv_1\cdot\vtau(\x_m)e^{-i\omega\vv_1\cdot\x_m}, & \vv_1\cdot\veta(\x_m)e^{-i\omega\vv_1\cdot\x_m} \\
          e^{-i\omega\vv_2\cdot\x_m}, & \vv_2\cdot\vtau(\x_m)e^{-i\omega\vv_2\cdot\x_m}, & \vv_2\cdot\veta(\x_m)e^{-i\omega\vv_2\cdot\x_m} \\
          \vdots & \vdots & \vdots \\
          e^{-i\omega\vv_P\cdot\x_m}, & \vv_P\cdot\vtau(\x_m)e^{-i\omega\vv_P\cdot\x_m}, & \vv_P\cdot\veta(\x_m)e^{-i\omega\vv_P\cdot\x_m} \\
        \end{array}
      \right]
\end{equation}
and
\begin{equation}\label{Vec2}
  \vH_m(\omega)=\left[
        \begin{array}{ccc}
          e^{i\omega\vt_1\cdot\x_m}, & \vt_1\cdot\vtau(\x_m)e^{i\omega\vt_1\cdot\x_m}, & \vt_1\cdot\veta(\x_m)e^{i\omega\vt_1\cdot\x_m} \\
          e^{i\omega\vt_2\cdot\x_m}, & \vt_2\cdot\vtau(\x_m)e^{i\omega\vt_2\cdot\x_m}, & \vt_2\cdot\veta(\x_m)e^{i\omega\vt_2\cdot\x_m} \\
          \vdots & \vdots & \vdots \\
          e^{i\omega\vt_Q\cdot\x_m}, & \vt_Q\cdot\vtau(\x_m)e^{i\omega\vt_Q\cdot\x_m}, & \vt_Q\cdot\veta(\x_m)e^{i\omega\vt_Q\cdot\x_m} \\
        \end{array}
      \right]^\mathrm{T},
\end{equation}
respectively.

The multi-frequency subspace migration imaging functional developed in \cite{AGKPS,P1,P2} is given hereafter; for several frequencies $\left\{\omega_k:k=1,2,\cdots,K\right\}$, Singular Value Decomposition (SVD) of $\mathbb{A}(\omega_k)$ is performed as follows
\begin{equation}\label{SVDofMSR}
  \mathbb{A}(\omega_k)=\mathbb{U}(\omega_k)\mathbb{S}(\omega_k)\mathbb{V}^*(\omega_k)\approx\sum_{m=1}^{M_k}\sigma_m(\omega_k)\vU_m(\omega_k)\vV_m^*(\omega_k),
\end{equation}
where $\sigma_m(\omega_k)$ are nonzero singular values, and $\vU_m(\omega_k)$ and $\vV_m(\omega_k)$ are respectively the left and right singular vectors of $\mathbb{A}(\omega_k)$. Then, based on the representations (\ref{Vec1}) and (\ref{Vec2}), for a \textit{suitable vector} $\va\in\mathbb{R}^3\backslash\{\mathbf{0}\}$ and $\z\in\Omega$, define vectors $\vE(\z;\omega_k)\in\mathbb{C}^{P\times1}$ and $\vF(\z;\omega)\in\mathbb{C}^{Q\times1}$ as
\begin{align}
\begin{aligned}\label{VecEF}
  \vE(\z;\omega_k)&=\bigg[\va\cdot[1,\vv_1]^\mathrm{T}e^{-i\omega\vv_1\cdot\z},\va\cdot[1,\vv_2]^\mathrm{T}e^{-i\omega\vv_2\cdot\z},\cdots, \va\cdot[1,\vv_P]^\mathrm{T}e^{-i\omega\vv_P\cdot\z}\bigg]^\mathrm{T}\\
  \vF(\z;\omega_k)&=\bigg[\va\cdot[1,\vt_1]^\mathrm{T}e^{i\omega\vt_1\cdot\z},\va\cdot[1,\vt_2]^\mathrm{T}e^{i\omega\vt_2\cdot\z},\cdots, \va\cdot[1,\vt_Q]^\mathrm{T}e^{i\omega\vt_Q\cdot\z}\bigg]^\mathrm{T},
\end{aligned}
\end{align}
respectively. Consequently, the corresponding unit vectors can be generated:
\[\vW_\mathrm{E}(\z;\omega_k):=\frac{\vE(\z;\omega_k)}{|\vE(\z;\omega_k)|}\quad\mbox{and}\quad \vW_\mathrm{F}(\z;\omega_k):=\frac{\vF(\z;\omega_k)}{|\vF(\z;\omega_k)|}.\]
Then, by comparing decompositions (\ref{SymmetryMSR}) and (\ref{SVDofMSR}), we can observe that following property holds (see \cite{AGKPS} also)
\begin{equation}\label{RelationSV}
  \vW_\mathrm{E}(\x_m;\omega_k)\simeq\vU_m(\omega_k)\quad\mbox{and}\quad \vW_\mathrm{F}(\x_m;\omega_k)\simeq\overline{\vV}_m(\omega_k),
\end{equation}
based on the structure of (\ref{SymmetryMSR}) and (\ref{SVDofMSR}), and the orthonormal property of singular vectors, we can examine the following:
\begin{align}
\begin{aligned}\label{OrthonormalVectors}
  \vW_\mathrm{E}^*(\z;\omega_k)\vU_m(\omega_k)\approx1&\quad\mbox{and}\quad\vW_\mathrm{F}^*(\z;\omega_k)\overline{\vV}_m(\omega_k)\approx1\quad\mbox{if}\quad\z=\x_m\\
  \vW_\mathrm{E}^*(\z;\omega_k)\vU_m(\omega_k)\approx0&\quad\mbox{and}\quad\vW_\mathrm{F}^*(\z;\omega_k)\overline{\vV}_m(\omega_k)\approx0\quad\mbox{if}\quad\z\ne\x_m.
\end{aligned}
\end{align}
Subsequently, we can design the following multi-frequency imaging functional (see \cite{AGKPS,P1,P2,PL2})
\begin{equation}\label{MultiFrequencyImagingFunction}
  \mathbb{W}(\z,K):=\left|\sum_{k=1}^{K}\sum_{m=1}^{M_k}\bigg(\vW_\mathrm{E}^*(\z;\omega_k)\vU_m(\omega_k)\bigg)\bigg(\vW_\mathrm{F}^*(\z;\omega_k)\overline{\vV}_m(\omega_k)\bigg)\right|,
\end{equation}
Then, based on (\ref{OrthonormalVectors}), imaging functional (\ref{MultiFrequencyImagingFunction}) should plot magnitude $1$ and $0$ at $\z=\x_m\in\Upsilon$ and $\z\ne\x_m$, respectively. Although this fact indicates why subspace migration imaging functional produces an image of thin inclusion(s), a mathematical analysis is still needed for explaining some phenomena such as the appearance of unexpected ghost replicas. In the following section, we analyze the imaging functional (\ref{MultiFrequencyImagingFunction}) and improve it for better imaging performance.

\section{Introduction to weighted multi-frequency imaging functional}\label{sec:3}
\subsection{Analysis of multi-frequency imaging functional (\ref{MultiFrequencyImagingFunction})}
Now, we analyze the multi-frequency imaging functional (\ref{MultiFrequencyImagingFunction}). To do this, we assume that $P$ is sufficiently large and number of nonzero singular values $M_k$ is almost equal to $M$, for $k=1,2,\cdots,K$. We then obtain the following results.

\begin{thm}\label{TheoremMultiFrequencyImaging}
  Assuming that $K$ is sufficiently large, then (\ref{MultiFrequencyImagingFunction}) becomes
  \begin{enumerate}
    \item If $Q\gg3M$ and $\omega_K<+\infty$, then
    \begin{equation}\label{StructureMFIF}
      \mathbb{W}(\z,K)\sim\frac{K}{\omega_K-\omega_1}\sum_{m=1}^{M}\left(\Phi(|\z-\x_m|;\omega_K)-\Phi(|\z-\x_m|;\omega_1)+\int_{\omega_1}^{\omega_K}J_1(\omega|\z-\x_m|)^2d\omega\right).
    \end{equation}
    Here, $\Phi(t;\omega)$ is given by
    \[\Phi(t;\omega):=\omega\bigg(J_0(\omega t)^2+J_1(\omega t)^2\bigg),\]
    where $J_\nu(t)$ denotes the Bessel function of order $\nu$ and of the first kind.
    \item If $Q\gg3M$ and $\omega_K\longrightarrow+\infty$, then
    \[\mathbb{W}(\z,K)\sim\delta(\z,\x_m),\]
    where $\delta$ denotes the Dirac delta function.
    \item If $Q>3M$ and $\omega_K\longrightarrow+\infty$, then
    \[\mathbb{W}(\z,K)\sim\frac{K}{\omega_K-\omega_1}\sum_{m=1}^{M}\sum_{q=1}^{Q}\frac{1}{\sqrt{|\z-\x_m|^2-\left(\vt_q\cdot(\z-\x_m)\right)^2}}.\]
  \end{enumerate}
\end{thm}
\begin{proof}
First, assume that $Q$ is sufficiently large such that $Q\gg 3M$ and that $\omega_K<+\infty$. We denote $\triangle\vv_p:=|\vv_{p}-\vv_{p-1}|$ for $p=2,3,\cdots,P$, and $\vv_1:=|\vv_{1}-\vv_{N}|$. Notation $\triangle\vt_q$ can be defined analogously.

Plugging the orthonormal relationship (\ref{OrthonormalVectors}) into (\ref{MultiFrequencyImagingFunction}) and applying \cite[Lemma 4.1]{G}, we observe that
\begin{align}
\begin{aligned}\label{ImagingFunctionBessel}
  \mathbb{W}(\z,K)&\sim\left|\sum_{k=1}^{K}\sum_{m=1}^{M}\bigg(\vW_\mathrm{E}^*(\z;\omega_k)\vW_\mathrm{E}(\x_m;\omega_k)\bigg)\bigg(\vW_\mathrm{F}^*(\z;\omega_k)\vW_\mathrm{F}(\x_m;\omega_k)\bigg)\right|\\
  &\approx\left|\sum_{k=1}^{K}\sum_{m=1}^{M}\left(\sum_{p=1}^{P}e^{i\omega_k\vv_p\cdot(\z-\x_m)}\right)\left(\sum_{q=1}^{Q}e^{i\omega_k\vt_q\cdot(\z-\x_m)}\right)\right|\\
  &\approx\left|\sum_{k=1}^{K}\sum_{m=1}^{M}\left(\frac{1}{2\pi}\sum_{p=1}^{P}e^{i\omega_k\vv_p\cdot(\z-\x_m)}\triangle\vv_p\right)\left(\frac{1}{2\pi}\sum_{q=1}^{Q}e^{i\omega_k\vt_q\cdot(\z-\x_m)}\triangle\vt_q\right)\right|\\
  &\approx\frac{1}{4\pi^2}\left|\sum_{k=1}^{K}\sum_{m=1}^{M}\left(\int_{\mathbb{S}^1}e^{i\omega_k\vt\cdot(\z-\x_m)}d\vt\right)^2\right|
  \approx\frac{K}{\omega_K-\omega_1}\left|\sum_{m=1}^{M}\int_{\omega_1}^{\omega_K}J_0(\omega|\z-\x_m|)^2d\omega\right|.
\end{aligned}
\end{align}
\begin{enumerate}
  \item Applying the indefinite integral (see \cite[Page 35]{R})
    \[\int J_0(t)^2dt=t\bigg(J_0(t)^2+J_1(t)^2\bigg)+\int J_1(t)^2dt\]
    with a change of variable $\omega|\z-\x_m|=t$ yields
    \begin{align}
    \begin{aligned}\label{StructureImaging}
      \mathbb{W}(\z,K)\sim&\frac{K}{\omega_K-\omega_1}\left|\sum_{m=1}^{M}\frac{1}{|\z-\x_m|}\int_{\omega_1|\z-\x_m|}^{\omega_K|\z-\x_m|}J_0(t)^2dt\right|\\ =&\frac{K}{\omega_K-\omega_1}\left|\sum_{m=1}^{M}\left\{\omega_K\bigg(J_0(\omega_K|\z-\x_m|)^2+J_1(\omega_K|\z-\x_m|)^2\bigg)+\int_{0}^{\omega_K}J_1(\omega|\z-\x_m|)^2d\omega\right.\right.\\
      &\left.\left.-\omega_1\bigg(J_0(\omega_1|\z-\x_m|)^2+J_1(\omega_1|\z-\x_m|)^2\bigg)-\int_{0}^{\omega_1}J_1(\omega|\z-\x_m|)^2d\omega\right\}\right|\\
      =&\frac{K}{\omega_K-\omega_1}\left|\sum_{m=1}^{M}\bigg(\Phi(|\z-\x_m|;\omega_K)-\Phi(|\z-\x_m|;\omega_1)\bigg)\right|.
    \end{aligned}
    \end{align}
  \item Let us assume $\omega_K\longrightarrow+\infty$. If $\z=\x_m$ then it is clear $\mathbb{W}(\z,K)=\infty$. Suppose that $\z\ne\x_m$ then the following asymptotic form of Bessel function holds for $\omega|\z-\x_m|\gg|\nu^2-0.25|$,
    \begin{equation}\label{AsymptoticBessel}
      J_{\nu}(\omega|\z-\x_m|)\approx\sqrt{\frac{2}{\omega\pi|\z-\x_m|}}\cos\left\{\omega|\z-\x_m|-\frac{\nu\pi}{2}-\frac{\pi}{4} +\mathcal{O}\left(\frac{1}{\omega|\z-\x_m|}\right)\right\},
    \end{equation}
    where $\nu$ denotes a positive integer. Applying this asymptotic form to (\ref{StructureImaging}), we can easily observe that $\mathbb{W}(\z,K)\approx0$. Hence, we conclude that
    \[\mathbb{W}(\z,K)\sim\delta(\z,\x_m).\]
  \item Suppose that $Q>3M$ but $Q\not\gg3M$ and $\omega_K\longrightarrow+\infty$. Then
  \begin{multline*}
    \mathbb{W}(\z,K)\approx\left|\sum_{k=1}^{K}\sum_{m=1}^{M}\left(\sum_{p=1}^{P}e^{i\omega_k\vv_p\cdot(\z-\x_m)}\right)\left(\sum_{q=1}^{Q}e^{i\omega_k\vt_q\cdot(\z-\x_m)}\right)\right|\\
    \sim\frac{K}{\omega_K-\omega_1}\left|\sum_{m=1}^{M}\sum_{q=1}^{Q}\int_{\omega_1}^{\omega_K}e^{i\omega\vt_q\cdot(\z-\x_m)}J_0(\omega|\z-\x_m|)d\omega\right|.
  \end{multline*}
    Since $\vt_q\in\mathbb{S}^1$, the following holds:
    \[|\z-\x_m|^2-\left(\vt_q\cdot(\z-\x_m)\right)^2=|\z-\x_m|^2\left\{1-\left(\vt_q\cdot\frac{\z-\x_m}{|\z-\x_m|}\right)^2\right\}\geq0.\]
    Hence, applying the identity (see \cite[formula 6.671, page 717]{GR})
    \[\int_0^\infty e^{iat}J_\nu(bt)dt=\frac{1}{\sqrt{b^2-a^2}}\left\{\cos\left(\nu\sin^{-1}\frac{a}{b}\right)+i\sin\left(\nu\sin^{-1}\frac{a}{b}\right)\right\}\quad\mbox{for}\quad a<b\]
    yields
    \[\lim_{\omega_K\to+\infty}\int_{\omega_1}^{\omega_K}e^{i\omega\vt_q\cdot(\z-\x_m)}J_0(\omega|\z-\x_m|)^2d\omega\approx \frac{1}{\sqrt{|\z-\x_m|^2-\left(\vt_q\cdot(\z-\x_m)\right)^2}}.\]
\end{enumerate}
\end{proof}

\begin{rem}\label{RemarkImagingFunction1}Theorem \ref{TheoremMultiFrequencyImaging} indicates certain properties of (\ref{MultiFrequencyImagingFunction}), which are summarized as follows:
\begin{enumerate}
  \item $J_0(x)$ has an oscillation property. Hence, if a small value of $K$ is considered, the identification of shape of $\Upsilon$ will be associated with ghost replicas.
  \item On the other hand, if a large number of frequencies are applied, $\mathbb{W}(\z,K)$ will exhibit an accurate shape of $\Upsilon$. Note that this fact can be validated via Statistical Hypothesis Testing (refer to \cite{AGKPS} for more details).
  \item When $Q$ is small, $\mathbb{W}(\z,K)$ plots a large magnitude at $\z$ satisfying
      \[\z=\x_m\in\Upsilon\quad\mbox{and}\quad\vt_q=\pm\frac{\z-\x_m}{|\z-\x_m|}.\]
      This means that $\mathbb{W}(\z,K)$ produces not only the shape of $\Upsilon$ but also unexpected ghost replicas.
  \item If low frequencies $\omega_k$ are applied such that
  \[\omega_k|\z-\x_m|\approx0,\]
  then $\mathbb{W}(\z,K)$ will fail to produce the shape of $\Upsilon$. However, if sufficiently high frequencies $\omega_k$ are applied, then following holds:
      \[\int_{\omega_1}^{\omega_K}J_1(\omega|\z-\x_m|)^2d\omega\ll \mathcal{O}(\omega_K).\]
      Therefore, the last term of (\ref{StructureMFIF}) can be negligible and hence the shape of $\Upsilon$ can be successfully imaged by $\mathbb{W}(\z,K)$. This is the reason why applying high frequency yields good results (see \cite{JKHP}).
  \item For a small value of $K$, if $\omega_K\longrightarrow+\infty$, then it is expected that $\mathbb{W}(\z,K)$ will yields a good result. However, this is an ideal assumption.
\end{enumerate}
\end{rem}

\subsection{Permeability contrast case: why do two ghost replicas appear in the numerical experiments?}\label{sec:3-2}
In section \ref{sec:2}, although we did not focus on the vector $\va$, it is important. In order to determine the influence of $\va$ on imaging performance, let us define vectors $\vE(\z;\omega_k)\in\mathbb{C}^{P\times1}$ and $\vF(\z;\omega)\in\mathbb{C}^{Q\times1}$ in (\ref{VecEF}) as
\begin{align}
\begin{aligned}\label{VecEFP}
  \vE(\z;\omega_k)&=\bigg[e^{-i\omega\vv_1\cdot\z},e^{-i\omega\vv_2\cdot\z},\cdots,e^{-i\omega\vv_P\cdot\z}\bigg]^\mathrm{T},\\
  \vF(\z;\omega_k)&=\bigg[e^{i\omega\vt_1\cdot\z},e^{i\omega\vt_2\cdot\z},\cdots,e^{i\omega\vt_Q\cdot\z}\bigg]^\mathrm{T},
\end{aligned}
\end{align}
and corresponding unit vectors $\vW_\mathrm{E}(\z;\omega_k)$ and $\vW_\mathrm{F}(\z;\omega_k)$. Then applying \cite[Lemma 4.1]{G}, $\mathbb{W}(\z,K)$ becomes
\begin{align*}
  \mathbb{W}(\z,K)\sim&\left|\sum_{k=1}^{K}\sum_{m=1}^{M}\bigg(\vW_\mathrm{E}^*(\z;\omega_k)\vU_m(\omega_k)\bigg)\bigg(\vW_\mathrm{F}^*(\z;\omega_k)\vV_m(\omega_k)\bigg)\right|\\
  \approx&\left|\sum_{k=1}^{K}\sum_{m=1}^{M}\left(\sum_{p=1}^{P}\vv_p\cdot(\vtau(\x_m)+\veta(\x_m))e^{i\omega_k\vv_p\cdot(\z-\x_m)}\right)\left(\sum_{q=1}^{Q}\vt_q\cdot(\vtau(\x_m)+\veta(\x_m))e^{i\omega_k\vt_q\cdot(\z-\x_m)}\right)\right|\\
  \approx&\left|\sum_{k=1}^{K}\sum_{m=1}^{M}\left(\int_{\mathbb{S}^1}\vt\cdot(\vtau(\x_m)+\veta(\x_m))e^{i\omega_k\vt\cdot(\z-\x_m)}d\vt\right)^2\right|\\ \sim&\left|\sum_{m=1}^{M}\int_{\omega_1}^{\omega_K}\left\{\left(\frac{\z-\x_m}{|\z-\x_m|}\cdot\bigg(\vtau(\x_m)+\veta(\x_m)\bigg)\right)J_1(\omega|\z-\x_m|)\right\}^2d\omega\right|.
\end{align*}

Note that $J_1(\omega x)$ is maximum at two points $x_1$ and $x_2$, and is symmetric with respect to $x$ (see Figure \ref{FunctionTheta}). This means that $\mathbb{W}(\z,K)$ plots $0$ on $\Upsilon$ and produces two ghost replicas in the neighborhood of $\Upsilon$ (see Figure \ref{FigureT1P}). To obtain a good result, on the basis of the structure of (\ref{perturbation}) and (\ref{VecEF}), $\va$ must be a linear combination of $\vtau(\x_m)$ and $\veta(\x_m)$. Unfortunately, we have no \textit{prior} information of $\Upsilon$. Hence, finding an optimal $\va$ is an interesting research topic.

\begin{figure}[!ht]
\begin{center}
\includegraphics[width=0.9\textwidth]{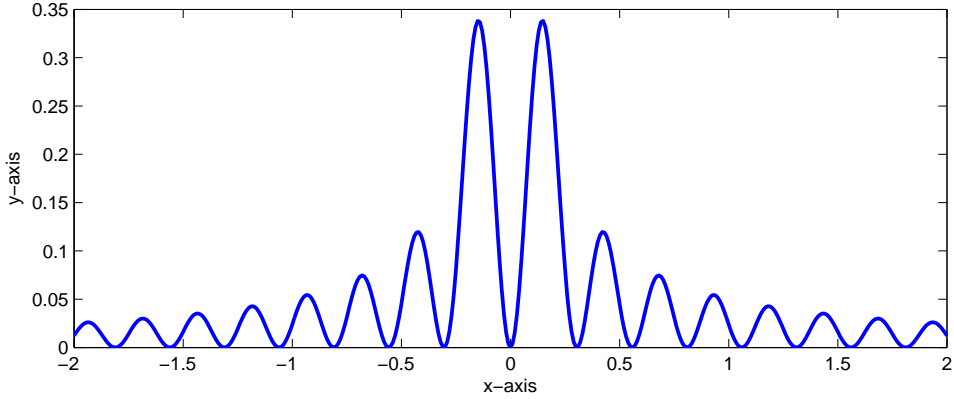}
\caption{\label{FunctionTheta}Graph of $y=J_1(\omega|x|)^2$ for $\omega=\frac{2\pi}{0.5}$.}
\end{center}
\end{figure}

\subsection{Weighted multi-frequency imaging functional: Introduction and analysis}
Remark \ref{RemarkImagingFunction1} shows an optimal condition ($P$ and $Q$ are sufficiently large and $\omega_K\longrightarrow+\infty$) for imaging. However, in real-world applications, a high frequency such as $+\infty$ cannot be applied. Therefore, eliminating the last term of (\ref{StructureMFIF}) should be a method of improvement. For this, following imaging functional weighted by $\omega_k$ is introduced in \cite{P5}:
\begin{equation}\label{WeightedMultiFrequencyImagingFunction1}
  \mathbb{W}(\z,K;1)=\left|\sum_{k=1}^{K}\sum_{m=1}^{M_k}\omega_k\bigg(\vW_\mathrm{E}^*(\z;\omega_k)\vU_m(\omega_k)\bigg)\bigg(\vW_\mathrm{F}^*(\z;\omega_k)\overline{\vV}_m(\omega_k)\bigg)\right|.
\end{equation}
This is an improved version of (\ref{MultiFrequencyImagingFunction}). Based on what is given above, it is natural to consider the following multi-frequency imaging functional weighted by $(\omega_k)^n$:
\begin{equation}\label{WeightedMultiFrequencyImagingFunction}
  \mathbb{W}(\z,K;n)=\left|\sum_{k=1}^{K}\sum_{m=1}^{M_k}(\omega_k)^n\bigg(\vW_\mathrm{E}^*(\z;\omega_k)\vU_m(\omega_k)\bigg)\bigg(\vW_\mathrm{F}^*(\z;\omega_k)\overline{\vV}_m(\omega_k)\bigg)\right|,\quad n\in\mathbb{N}.
\end{equation}
Consequently, the following question arises:
\begin{quote}
\textit{Does a large value of order $n$ significantly contributes to the imaging performance?}
\end{quote}
The answer is \textit{no}. To find out how, the structure of (\ref{WeightedMultiFrequencyImagingFunction}) is explored in the following Theorem.

\begin{thm}\label{TheoremStructureWMFIF}
  Let $K$ and $\omega_K$ have sufficiently large values and $\omega_K<+\infty$. Then for a natural number $n$, (\ref{WeightedMultiFrequencyImagingFunction}) satisfies
  \begin{equation}\label{StructureWMFIF}
    \mathbb{W}(\z,K;n)\sim \frac{K}{\omega_K-\omega_1}\sum_{m=1}^{M}\bigg(\hat{\Phi}(|\z-\x_m|,\omega_K;n)-\hat{\Phi}(|\z-\x_m|,\omega_1;n)\bigg),
  \end{equation}
  where $\hat{\Phi}$ is given by
  \[\hat{\Phi}(|\z-\x_m|,\omega;n):=\omega^{n+1}\bigg(J_0(\omega|\z-\x_m|)^2+J_1(\omega|\z-\x_m|)^2\bigg).\]
\end{thm}
\begin{proof} First, let us assume that $n$ is an odd number, say $n=2s+1$. In order to examine the structure of (\ref{WeightedMultiFrequencyImagingFunction}), we recall a recurrence formula (see \cite[page 14]{R})
\begin{align}
\begin{aligned}\label{Recurrence1}
  \int t^{2n+1}J_0(t)^2dt=\frac{1}{2n+1}\left\{\frac{t^{2n+2}}{2}\bigg(J_0(t)^2+J_1(t)^2\bigg)
  +n^2t^{2n}J_0(t)^2+nt^{2n+1}J_0(t)J_1(t)-2n^3\int t^{2n-1}J_0(t)^2dt\right\}.
\end{aligned}
\end{align}
Then, applying (\ref{Recurrence1}) with a change of variable $\omega|\z-\x_m|=t$ in (\ref{WeightedMultiFrequencyImagingFunction}) induces
\begin{align*}
  \mathbb{W}(\z,K;2s+1)\sim \frac{K}{\omega_K-\omega_1}&\left|\sum_{m=1}^{M}\frac{1}{|\z-\x_m|^{2s+2}}\int_{\omega_1|\z-\x_m|}^{\omega_K|\z-\x_m|}t^{2s+1}J_0(t)^2dt\right|\\ =\frac{K}{\omega_K-\omega_1}&\left|\sum_{m=1}^{M}\bigg(\Psi_1(|\z-\x_m|,\omega_K,\omega_1;s)+\Psi_2(|\z-\x_m|,\omega_K,\omega_1;s)\right.\\
  &\left.+\Psi_3(|\z-\x_m|,\omega_K,\omega_1;s)-\Psi_4(|\z-\x_m|,\omega_K,\omega_1;s)\bigg)\right|,
\end{align*}
where
\begin{align}\label{FunctionPsi}
\begin{aligned}
  \Psi_1(t,\omega_K,\omega_1;s)&:=\frac{(\omega_K)^{2s+2}}{4s+2}\bigg(J_0(\omega_Kt)^2+J_1(\omega_Kt)^2\bigg)-\frac{(\omega_1)^{2s+2}}{4s+2}\bigg(J_0(\omega_1t)^2+J_1(\omega_1t)^2\bigg),\\
  \Psi_2(t,\omega_K,\omega_1;s)&:=\frac{s^2}{(2s+1)t^2}\bigg((\omega_K)^{2s}J_0(\omega_Kt)^2-(\omega_1)^{2s}J_0(\omega_1t)^2\bigg),\\
  \Psi_3(t,\omega_K,\omega_1;s)&:=\frac{s(\omega_K)^{2s+1}}{(2s+1)t}J_0(\omega_Kt)J_1(\omega_Kt)-\frac{s(\omega_1)^{2s+1}}{(2s+1)t}J_0(\omega_1t)J_1(\omega_1t),
\end{aligned}
\end{align}
and
\[\Psi_4(|\z-\x_m|,\omega_K,\omega_1;s):=\frac{2s^3}{(2s+1)|\z-\x_m|^2}\mathbb{W}(\z,K;2s-1).\]

Now, we assume that $\omega_K$ is sufficiently large, $\z$ satisfies $\z\ne\x_m$ and $|\z-\x_m|\gg\frac{3}{4\omega_K}$. Then, applying the asymptotic form of the Bessel function (\ref{AsymptoticBessel}) yields
\[\frac{J_0(\omega_K|\z-\x_m|)^2}{|\z-\x_m|^2}\approx\frac{2}{\omega_K|\z-\x_m|^3\pi} \cos\bigg(\omega_K|\z-\x_m|-\frac{\pi}{4}\bigg)\ll\frac{128(\omega_K)^2}{27\pi} \cos\bigg(\omega_K|\z-\x_m|-\frac{\pi}{4}\bigg),\]
and
\begin{align*}
  \frac{J_0(\omega_K|\z-\x_m|)J_1(\omega_K|\z-\x_m|)}{|\z-\x_m|}&\approx \frac{2}{\omega_K|\z-\x_m|^2\pi}\cos\bigg(\omega_K|\z-\x_m|-\frac{\pi}{4}\bigg) \cos\bigg(\omega_K|\z-\x_m|-\frac{3\pi}{4}\bigg)\\
  &\ll\frac{32\omega_K}{9\pi}\cos\bigg(\omega_K|\z-\x_m|-\frac{\pi}{4}\bigg) \cos\bigg(\omega_K|\z-\x_m|-\frac{3\pi}{4}\bigg).
\end{align*}
Hence, we can observe that
\begin{align*}
  \Psi_2(|\z-\x_m|,\omega_K,\omega_1;s)&\ll \mathcal{O}\left((\omega_K)^{2s+2}\right)=\Psi_1(|\z-\x_m|,\omega_K,\omega_1;s),\\
  \Psi_3(|\z-\x_m|,\omega_K,\omega_1;s)&\ll \mathcal{O}\left((\omega_K)^{2s+2}\right)=\Psi_1(|\z-\x_m|,\omega_K,\omega_1;s),\\
  \Psi_4(|\z-\x_m|,\omega_K,\omega_1;s)&\ll\frac{32s^3(\omega_K)^2}{9(2s+1)}\mathbb{W}(\z,K;2s-1)=\mathcal{O}\left((\omega_K)^{2s+1}\right).
\end{align*}
This means that the terms $\Psi_2$, $\Psi_3$, and $\Psi_4$ are dominated by $\Psi_1$. Hence, we conclude that for an odd number $n$, (\ref{StructureWMFIF}) holds.

Next, let us assume that $n$ is an even number, say $n=2s$. Similar to the case of an odd number $n$, we recall a recurrence formula (see \cite[page 35]{R})
\begin{equation}\label{Recurrence2}
  \int t^{2n}J_0(t)^2dt=A_n(t)J_0(t)^2+B_n(t)J_0(t)J_1(t)+C_n(t)J_1(t)^2+D_n\int J_0(t)^2dt,
\end{equation}
where
\[A_n(t)=\sum_{r=1}^{n+1}a_r t^{2r-1},\quad B_n(t)=\sum_{r=1}^{n}b_r t^{2r},\quad C_n(t)=\sum_{r=1}^{n}c_nt^{2r+1},\]
and $a_r,b_r,c_r,$ and $D_n$ are constants. Then, applying a change of variable $t=\omega|\z-\x_m|$ in (\ref{Recurrence2}) yields
\begin{align*}
  \mathbb{W}(\z,K;2s)\sim \frac{K}{\omega_K-\omega_1}&\left|\sum_{m=1}^{M}\frac{1}{|\z-\x_m|^{2s+1}}\int_{\omega_1|\z-\x_m|}^{\omega_K|\z-\x_m|}t^{2s}J_0(t)^2dt\right|\\
  =\frac{K}{\omega_K-\omega_1}&\left|\sum_{m=1}^{M}\bigg(\Psi_5(|\z-\x_m|,\omega_K,\omega_1;s)+\Psi_6(|\z-\x_m|,\omega_K,\omega_1;s)\right.\\
  &\left.+\Psi_7(|\z-\x_m|,\omega_K,\omega_1;s)-\Psi_8(|\z-\x_m|,\omega_K,\omega_1;s)\bigg)\right|,
\end{align*}
where
\begin{align*}
  \Psi_5(t,\omega_K,\omega_1;s)&:=\frac{1}{t^{2s+1}}\sum_{r=1}^{s+1}\bigg(a_r(\omega_Kt)^{2r-1}J_0(\omega_Kt)^2 -a_r(\omega_1t)^{2r-1}J_0(\omega_1t)^2\bigg),\\
  \Psi_6(t,\omega_K,\omega_1;s)&:=\frac{1}{t^{2s+1}}\sum_{r=1}^{s}\bigg(b_r(\omega_Kt)^{2r}J_0(\omega_Kt)J_1(\omega_Kt) -b_r(\omega_1t)^{2r}J_0(\omega_1t)J_1(\omega_1t)\bigg),\\
  \Psi_7(t,\omega_K,\omega_1;s)&:=\frac{1}{t^{2s+1}}\sum_{r=1}^{s}\bigg(c_r(\omega_Kt)^{2r+1}J_1(\omega_Kt)^2 -c_r(\omega_1t)^{2r+1}J_1(\omega_1t)^2\bigg),
\end{align*}
and
\[\Psi_8(|\z-\x_m|,\omega_K,\omega_1;s)=\frac{D_n}{|\z-\x_m|^{2s}}\mathbb{W}(\z,K).\]

Now, we assume that $\omega_K$ is sufficiently large, $\z$ satisfies $\z\ne\x_m$, and $|\z-\x_m|\gg\frac{3}{4\omega_K}$. Then applying (\ref{AsymptoticBessel}) yields
\[\sum_{r=1}^{s}a_r(\omega_K|\z-\x_m|)^{2r-1}\ll\sum_{r=1}^{s}a_r\left(\frac{4}{3}\right)^{2r-1}\leq\max_{1\leq r\leq s}\{a_r\}\frac{12\cdot (4^{2s}-3^{2s})}{7\cdot 3^{2s}}.\]
Accordingly, we obtain
\begin{align*}
  \Psi_5(|\z-\x_m|,\omega_K,\omega_1;s)&\sim\bigg((\omega_K)^{2s+1}J_0(\omega_K|\z-\x_m|)^2-(\omega_1)^{2s+1}J_0(\omega_1|\z-\x_m|)^2\bigg),\\
  \Psi_6(|\z-\x_m|,\omega_K,\omega_1;s)&\ll\Psi_5(|\z-\x_m|,\omega_K,\omega_1;s),\\
  \Psi_7(|\z-\x_m|,\omega_K,\omega_1;s)&\sim\bigg((\omega_K)^{2s+1}J_1(\omega_K|\z-\x_m|)^2-(\omega_1)^{2s+1}J_1(\omega_1|\z-\x_m|)^2\bigg),\\
  \Psi_8(|\z-\x_m|,\omega_K,\omega_1;s)&\ll\left(\frac{4\omega_K}{3}\right)^{2s}D_n\mathbb{W}(\z,K)=\mathcal{O}\left((\omega_K)^{2s+1}\right).
\end{align*}
Hence, we conclude that for an even number $n$, (\ref{StructureWMFIF}) holds.
\end{proof}

\begin{rem}\label{RemarkImagingFunction2}Based on Theorem \ref{TheoremStructureWMFIF}, we can determine following properties of (\ref{WeightedMultiFrequencyImagingFunction}), as summarized below:
\begin{enumerate}\renewcommand{\theenumi}{P\arabic{enumi}}
  \item\label{P1Remark2} Increasing order $n$ in (\ref{WeightedMultiFrequencyImagingFunction}) does not influence the imaging performance. In fact, $\mathbf{n=1}$ \textbf{is the best choice for obtaining good results} because the terms $\Psi_l(|\z-\x_m|,\omega_K,\omega_1;s)$, $l=2,3,\cdots,8,$ in (\ref{WeightedMultiFrequencyImagingFunction}) completely disappear.
  \item Similar to (\ref{MultiFrequencyImagingFunction}), applying low frequencies $\omega_k$ such that $\omega_k|\z-\x_m|\approx0$ degrades the results. Hence, sufficiently high frequencies must be applied.
\end{enumerate}
\end{rem}

\section{Numerical experiments}\label{sec:4}
In this section, some numerical experiments are performed and corresponding results are exhibited. Throughout this section, the homogeneous domain $\Omega$ is selected as a unit circle and two $\gamma_j$ are chosen to describe the supporting curves of thin inclusions $\Upsilon_j$ as
\begin{align*}
  \gamma_1&=\left\{[z-0.2,-0.5z^2+0.4]^\mathrm{T}:-0.5\leq z\leq0.5\right\}\\
  \gamma_2&=\left\{[z+0.2,z^3+z^2-0.5]^\mathrm{T}:-0.5\leq z\leq0.5\right\}.
\end{align*}

The thickness of all the thin inclusion(s) is set to $h=0.015$, $\eps_0=\mu_0=1$ and $\eps_j=\mu_j=5$. For given wavelengths $\lambda_k$, the applied frequencies are $\omega_k=\frac{2\pi}{\lambda_k}$, which vary between $\omega_1=\frac{2\pi}{0.7}$ and $\omega_{K}=\frac{2\pi}{0.3}$. In every example, $K=10$ different frequencies are applied.

A set of vectors $\left\{\vv_j:j=1,2,\cdots,P\right\}$ and $\left\{\vt_l:l=1,2,\cdots,Q\right\}$ on $\mathbb{S}^1$ are selected as
\[\vv_j=\left[\cos\frac{2\pi j}{P},\sin\frac{2\pi j}{P}\right]^\mathrm{T}\quad\mbox{and}\quad
\vt_l=-\left[\cos\frac{2\pi l}{Q},\sin\frac{2\pi l}{Q}\right]^\mathrm{T},\]
respectively. In our examples, we set, $P=24$ and $Q=20$. Vector $\va$ in (\ref{VecEF}) is set to $\va=[1,0,1]^\mathrm{T}$.

For every example, white Gaussian noise with a Signal-to-Noise Ratio (SNR) of $10$dB is added to the unperturbed boundary measurement data using the MATLAB command \texttt{awgn} in order to demonstrate the effectiveness of the proposed algorithm. In order to discriminate non-zero singular values, a $0.01-$threshold strategy is adopted (see \cite{PL1,PL3} for instance).

First, let us consider the imaging of $\Upsilon_1$ when only magnetic permeability contrast exists, at a fixed frequency $\omega=\frac{2\pi}{0.5}$. From the result in section \ref{sec:3-2}, it is clear that two ghost replicas appear in the neighborhood of $\Upsilon_1$ when we select vectors $\vE(\z;\omega_k)$ and $\vF(\z;\omega)$ in (\ref{VecEFP}). Note that based on the shape of $\Upsilon_1$, $\vtau(\x_m)\approx0$ and $\veta(\x_m)\approx1$ for all $m=1,2,\cdots,M$. Hence setting of $\va=[0,0,1]^\mathrm{T}$ yields a good result while $\va=[0,1,0]^\mathrm{T}$ offers a poor result, as shown in Figure \ref{FigureT1P}.

\begin{figure}[!ht]
\begin{center}
\includegraphics[width=0.325\textwidth]{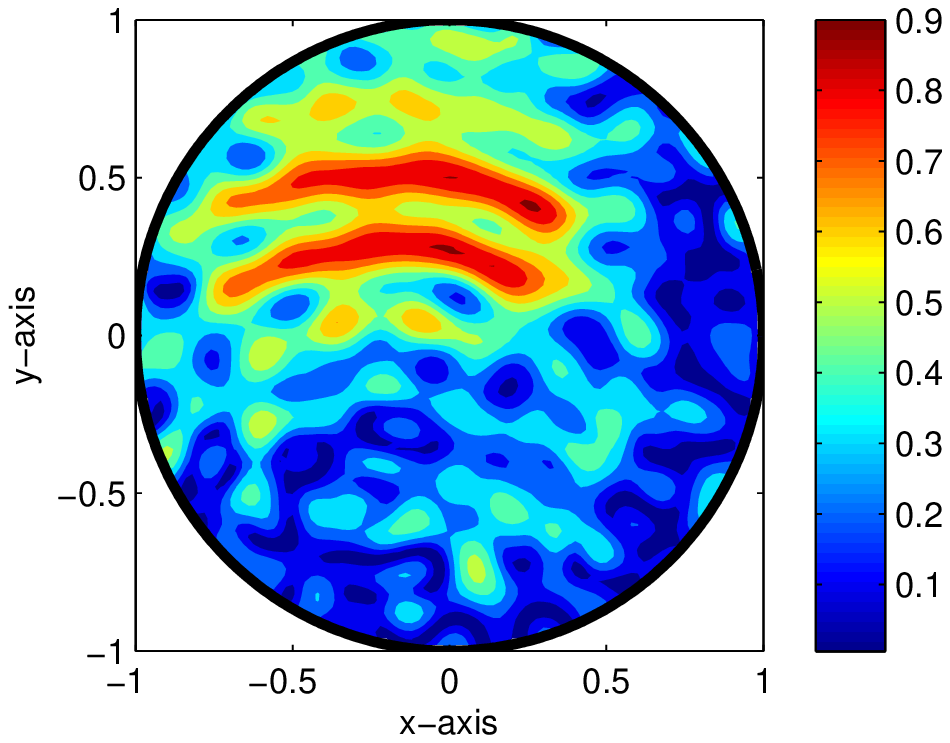}
\includegraphics[width=0.325\textwidth]{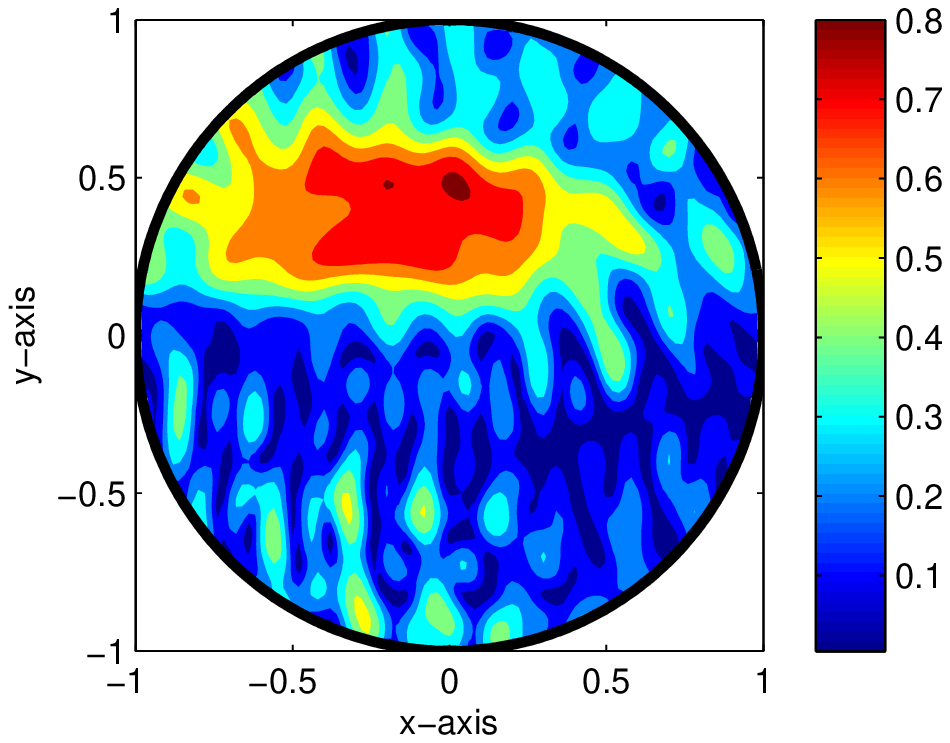}
\includegraphics[width=0.325\textwidth]{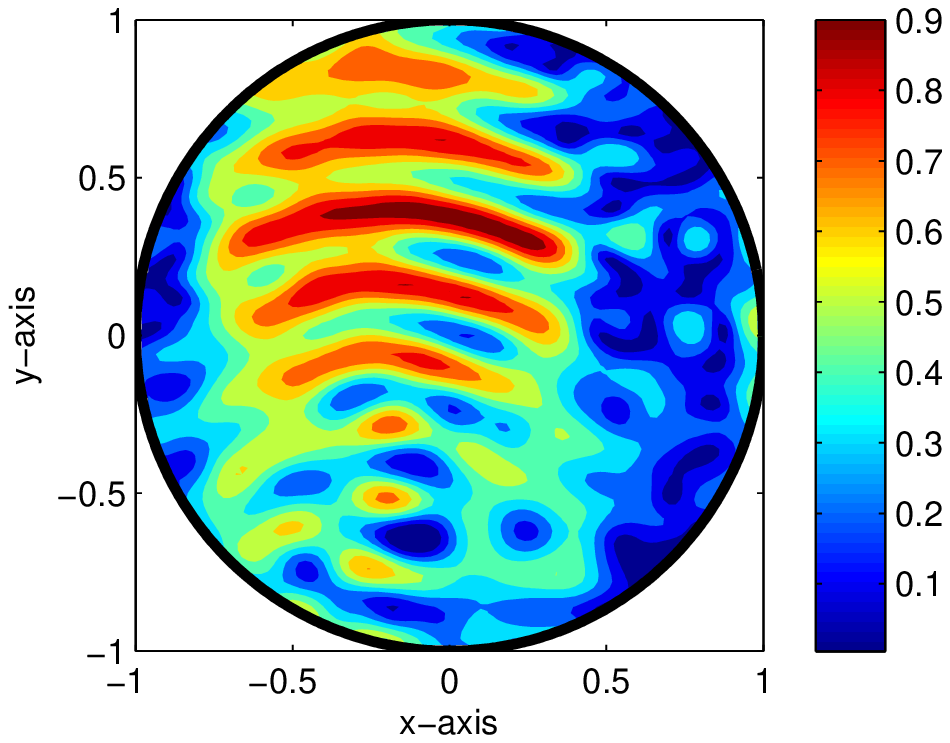}
\caption{\label{FigureT1P}[Permeability contrast case only] Maps of $\mathbb{W}(\z,0)$ with vectors in (\ref{VecEFP}) (left), $\va=[0,1,0]^\mathrm{T}$ (center), and $\va=[0,0,1]^\mathrm{T}$ (right) when the thin inclusion is $\Upsilon_1$.}
\end{center}
\end{figure}

Figure \ref{FigureT1PM} shows the effect of the number of frequencies on the imaging performance. By comparing the results with those in Figure \ref{FigureT1P}, we observe that applying multiple frequencies yields better results than applying a single frequency does.

\begin{figure}[!ht]
\begin{center}
\includegraphics[width=0.325\textwidth]{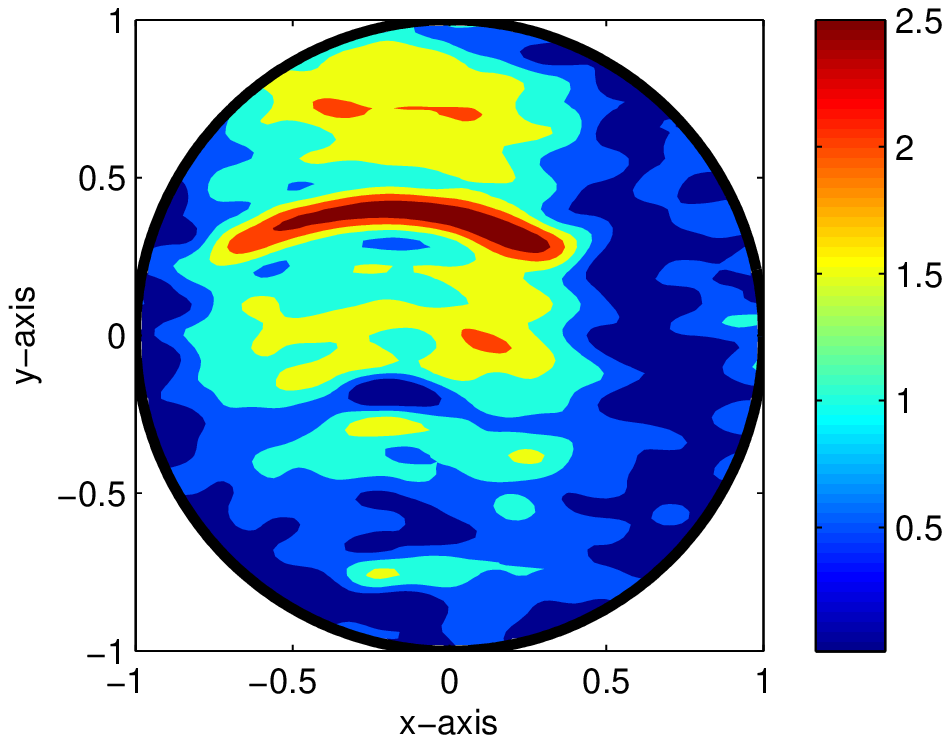}
\includegraphics[width=0.325\textwidth]{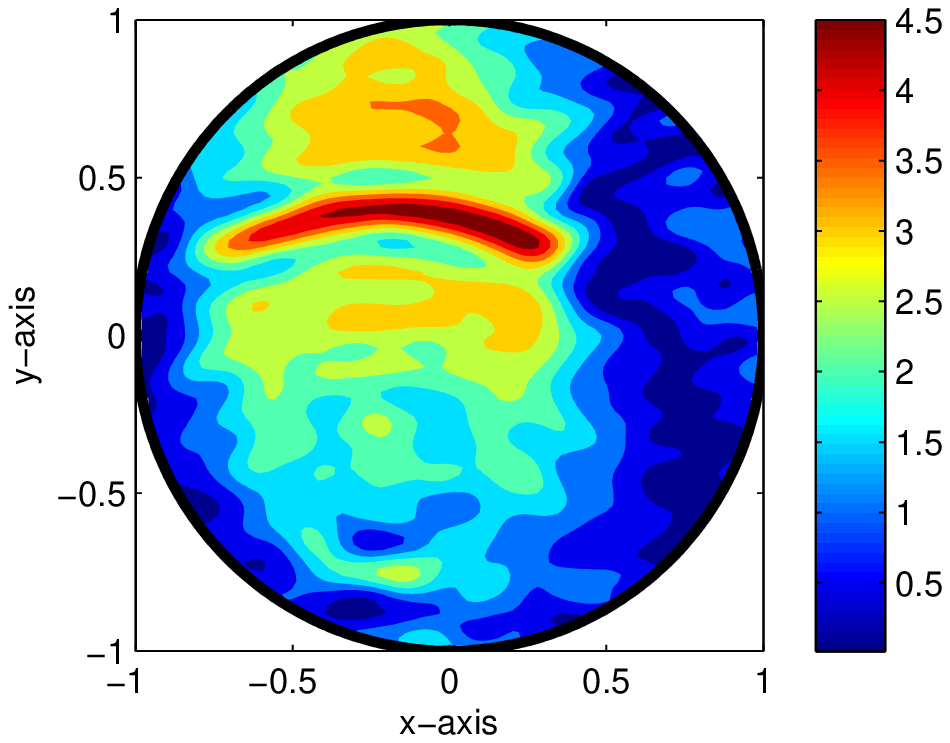}
\includegraphics[width=0.325\textwidth]{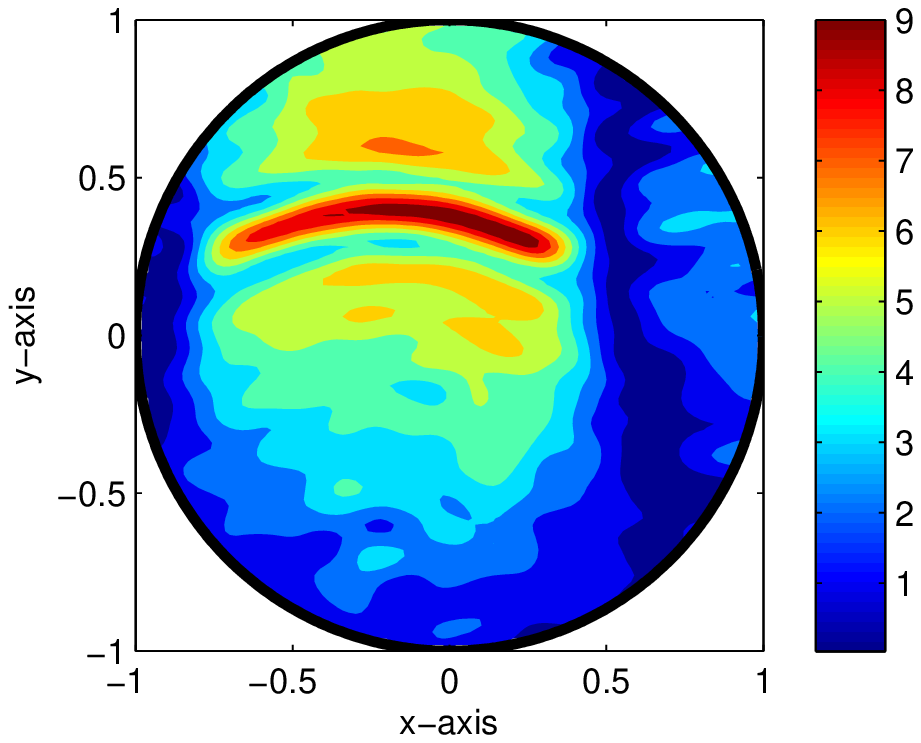}
\caption{\label{FigureT1PM}[Permeability contrast case only] Maps of $\mathbb{W}(\z,n)$ with $\va=[0,0,1]^\mathrm{T}$ for $n=3$ (left), $n=5$ (center), and $n=10$ (right) when the thin inclusion is $\Upsilon_1$.}
\end{center}
\end{figure}

Hereafter, we consider the imaging of thin inclusions when both permittivity and permeability contrast exist. Figure \ref{FigureT1} shows the maps of $\mathbb{W}(\z,10;n)$ for $n=0,1,$ and $2$. The overall results show that although $\mathbb{W}(\z,10;n)$ offers good imaging performance, $n=1$ yields better results than $n=0,2$ does, refer to (\ref{P1Remark2}) of Remark \ref{RemarkImagingFunction2}.

\begin{figure}[!ht]
\begin{center}
\includegraphics[width=0.325\textwidth]{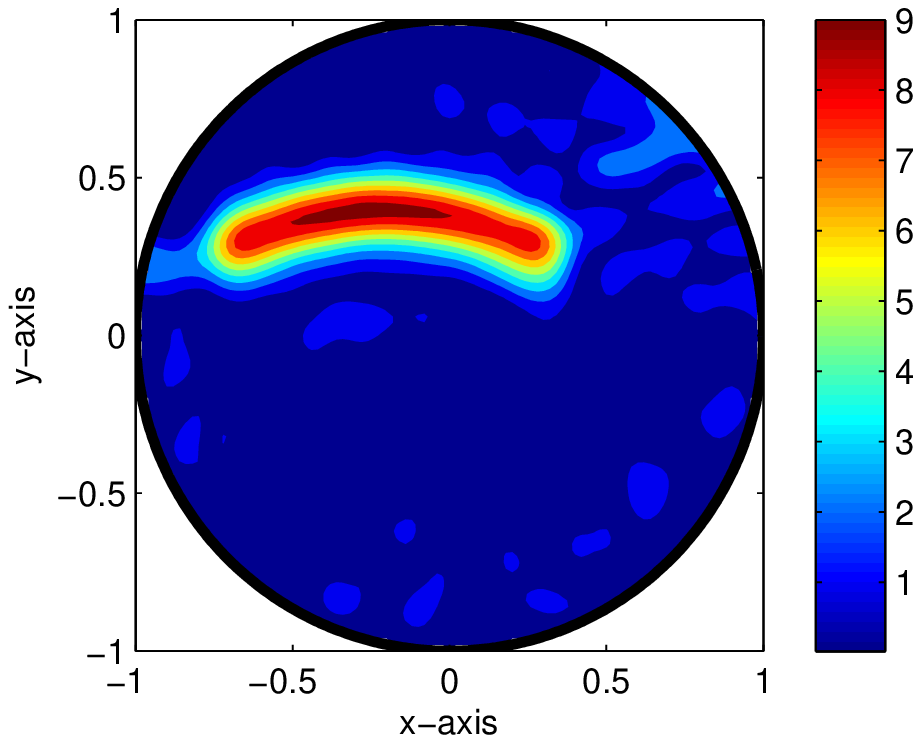}
\includegraphics[width=0.325\textwidth]{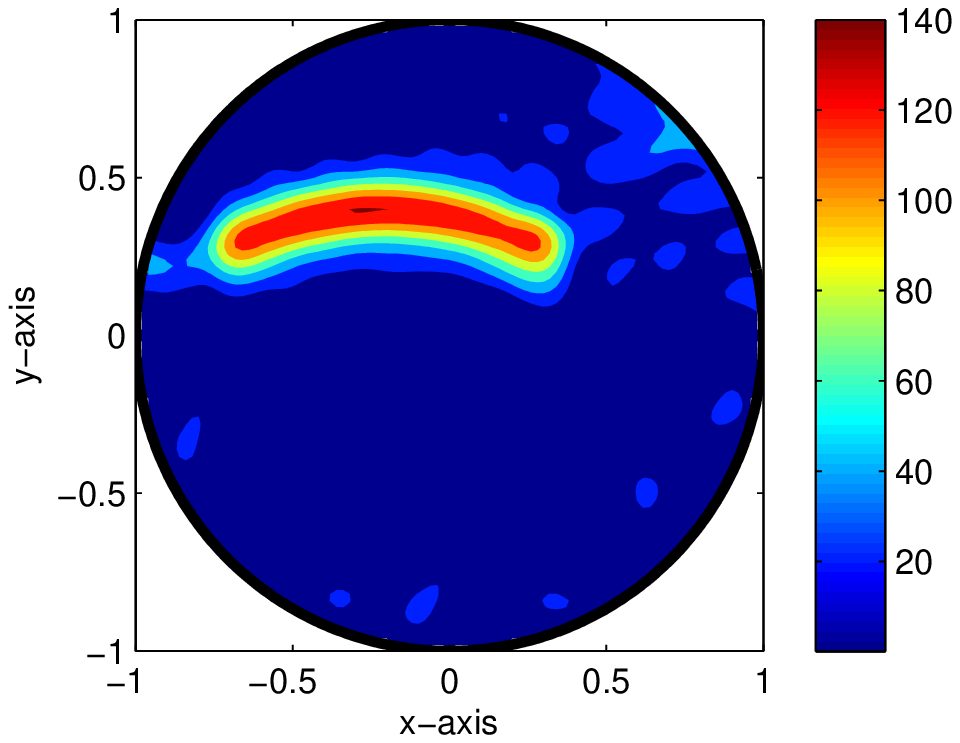}
\includegraphics[width=0.325\textwidth]{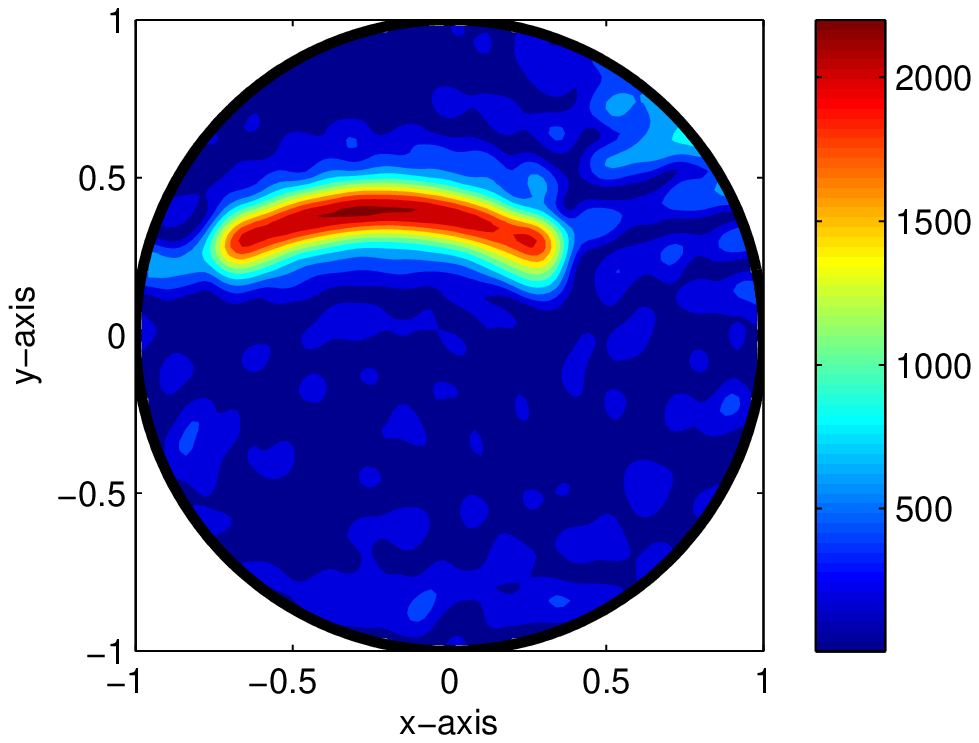}
\caption{\label{FigureT1}Maps of $\mathbb{W}(\z,10;n)$ for $n=0$ (left), $n=1$ (center), and $n=2$ (right) when the thin inclusion is $\Upsilon_1$.}
\end{center}
\end{figure}

Now, we consider the maps of $\mathbb{W}(\z,10;n)$ for $n=5,6,$ and $7$. The corresponding results illustrated in Figure \ref{FigureT1High} indicate that increasing $n$ not only contributes negligibly to the imaging performance but also generates large numbers of unexpected replicas. This result supports (\ref{P1Remark2}) of Remark \ref{RemarkImagingFunction2}.

\begin{figure}[!ht]
\begin{center}
\includegraphics[width=0.325\textwidth]{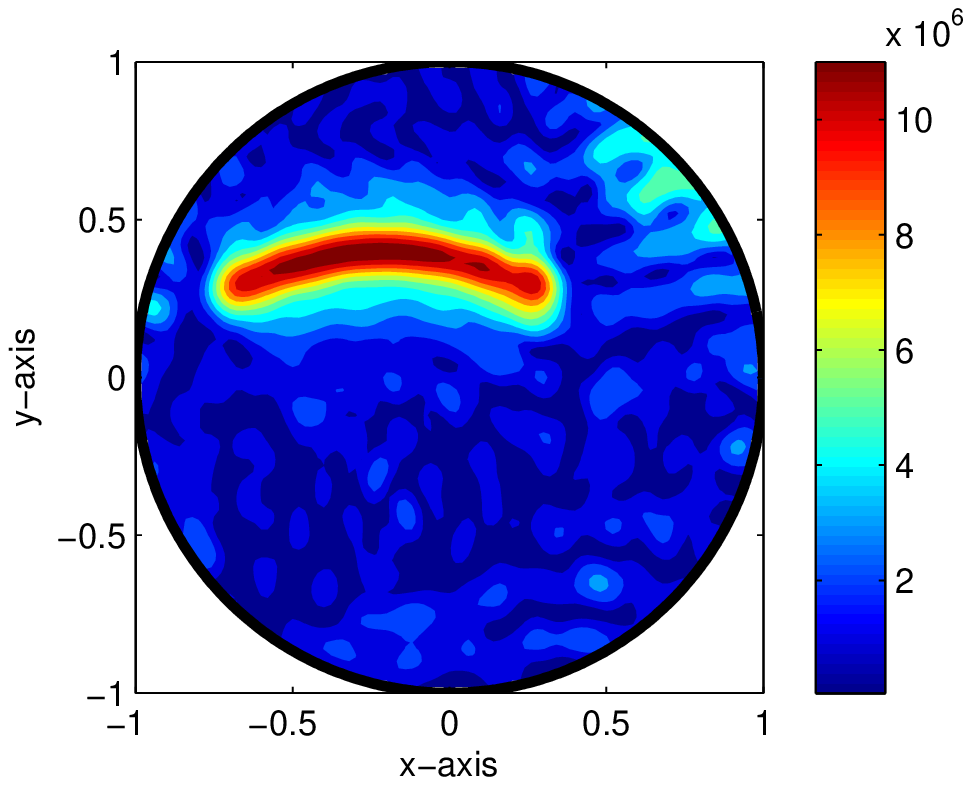}
\includegraphics[width=0.325\textwidth]{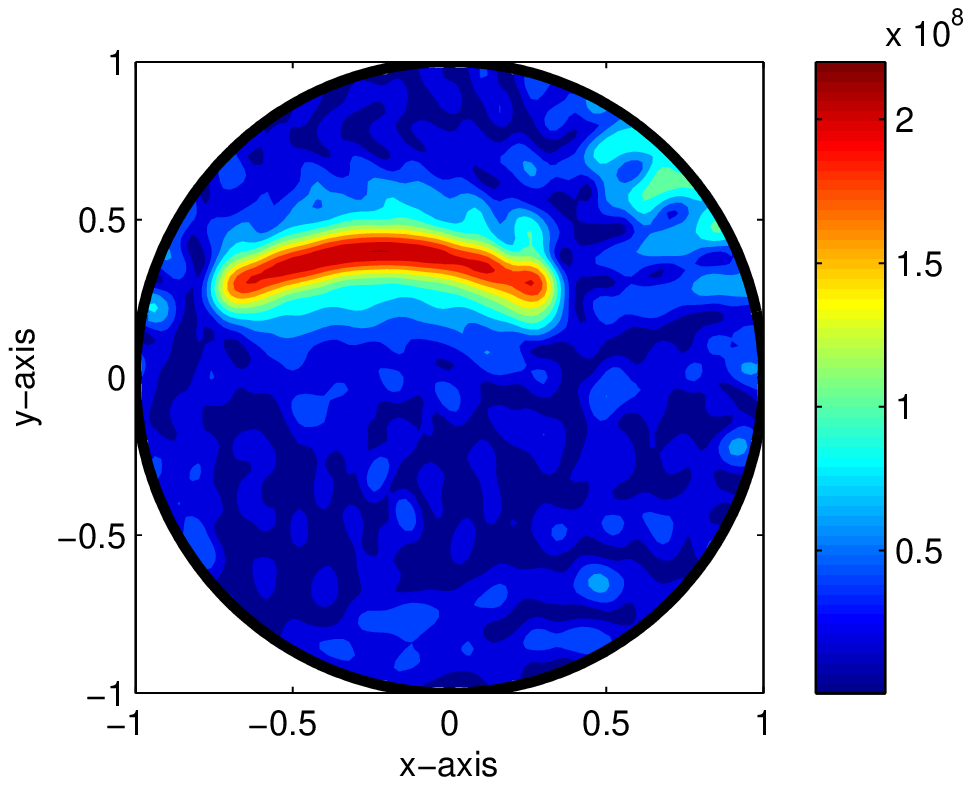}
\includegraphics[width=0.325\textwidth]{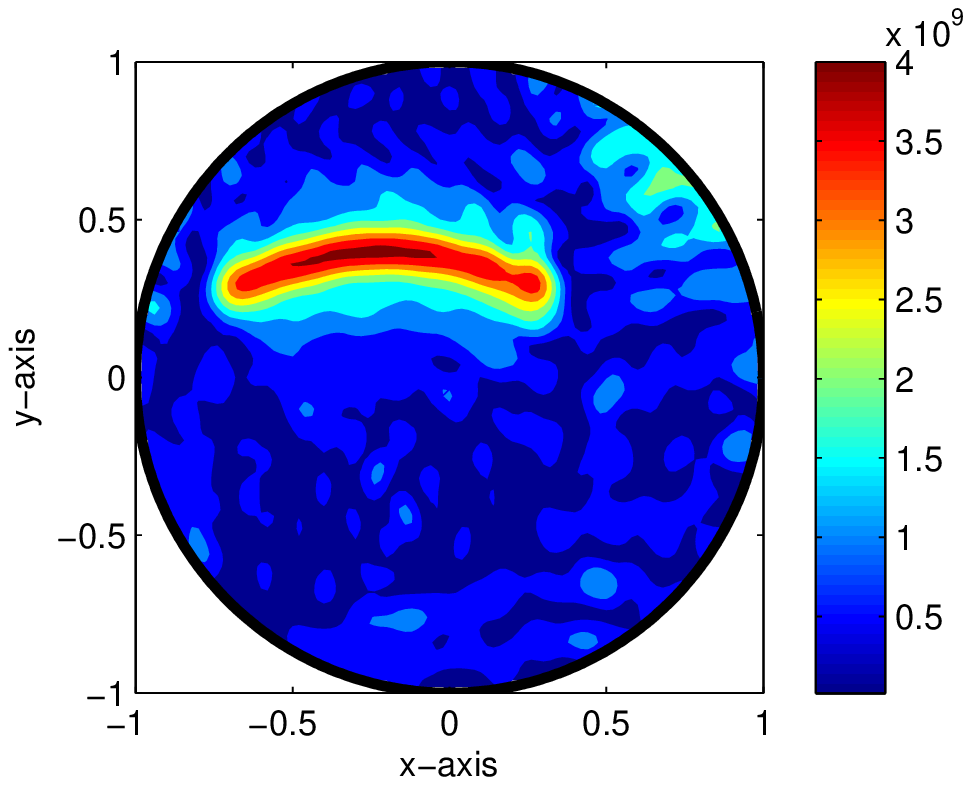}
\caption{\label{FigureT1High}Maps of $\mathbb{W}(\z,10;n)$ for $n=5$ (left), $n=6$ (center), and $n=7$ (right) when the thin inclusion is $\Upsilon_1$.}
\end{center}
\end{figure}

Figure \ref{FigureT2} shows the maps of $\mathbb{W}(\z,10;n)$ for $n=0,1,$ and $2$ when the thin inclusion is $\Upsilon_2$. Similar to Figure \ref{FigureT1}, the map of $\mathbb{W}(\z,10;1)$ demonstrates very good imaging performance, while the resolution is poor at $\z=[0.6,-0.24]^\mathrm{T}$ on the large curvature.

\begin{figure}[!ht]
\begin{center}
\includegraphics[width=0.325\textwidth]{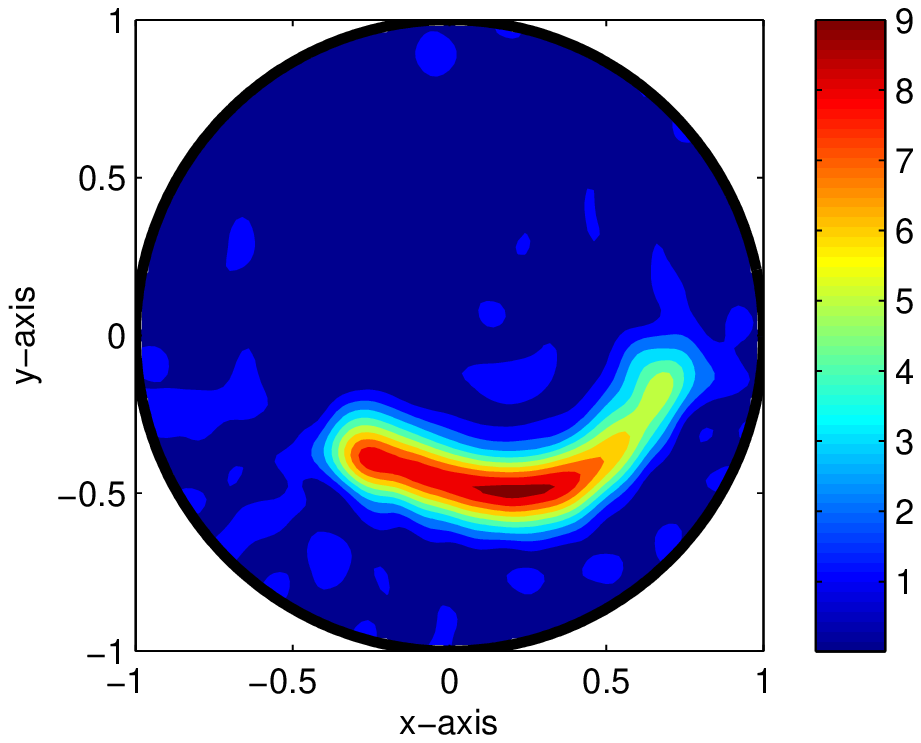}
\includegraphics[width=0.325\textwidth]{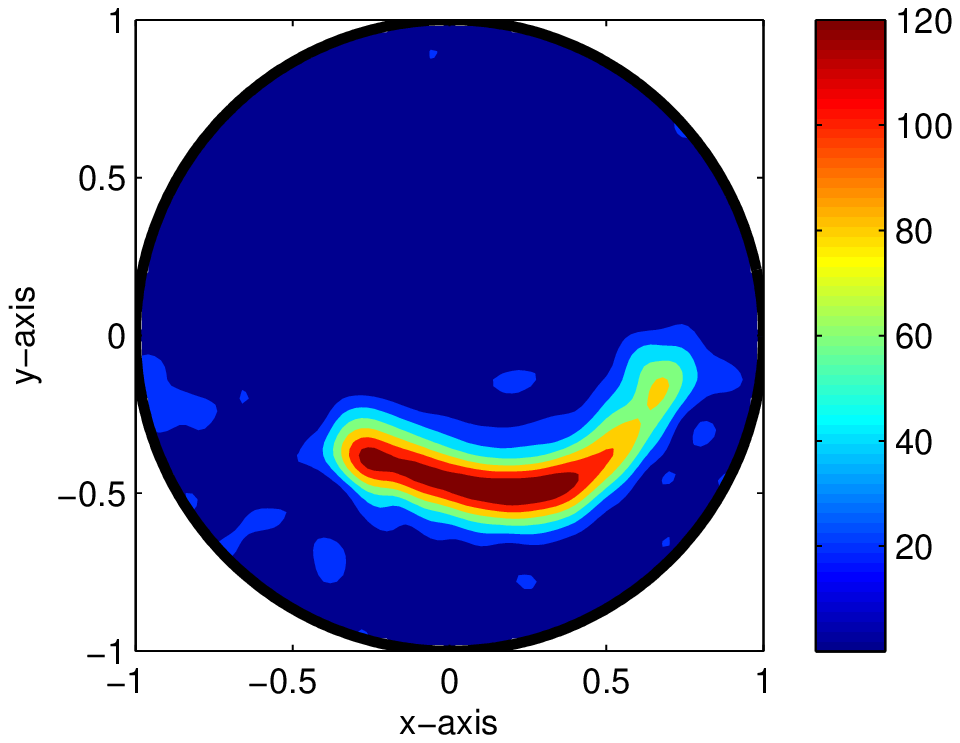}
\includegraphics[width=0.325\textwidth]{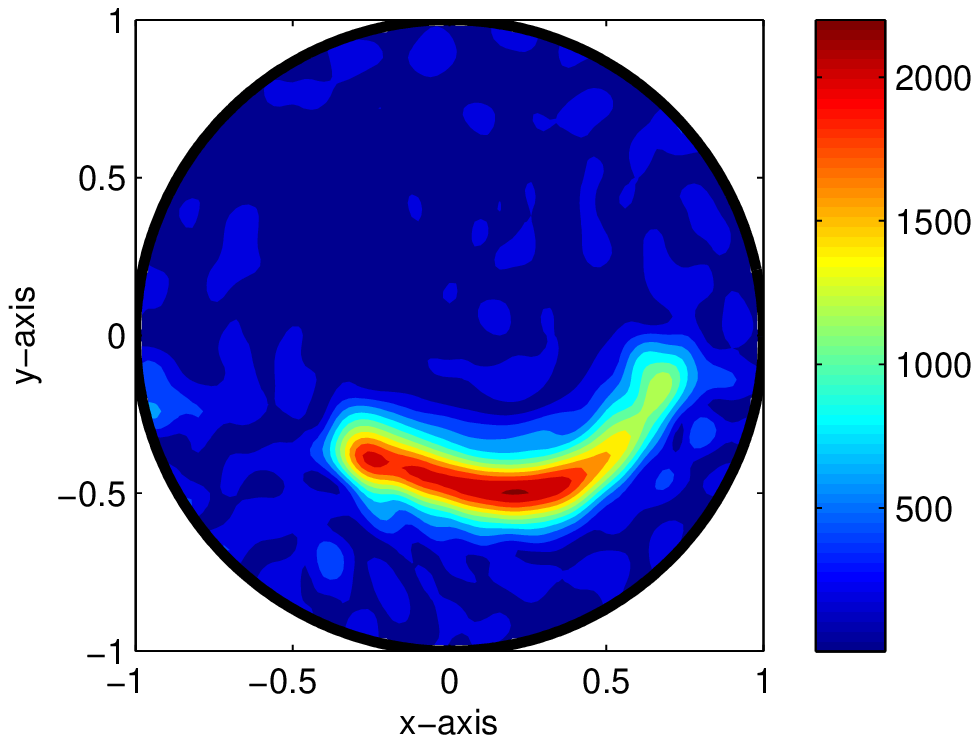}
\caption{\label{FigureT2}Same as Figure \ref{FigureT1} except that the thin inclusion is $\Upsilon_2$.}
\end{center}
\end{figure}

For the final example, we consider imaging of multiple inclusions $\Upsilon_1\cup\Upsilon_2$ with $\eps_1=\mu_1=5$ and $\eps_2=\mu_2=10$. Figure \ref{FigureTM} shows the maps of $\mathbb{W}(\z,10;n)$ for $n=0,1,$ and $2$. Similar to the previous examples, we can see that the map of $\mathbb{W}(\z,10;1)$ yields the most accurate shape of $\Upsilon_1\cup\Upsilon_2$, but because $\Upsilon_1$ has a much smaller values of permittivity and permeability than $\Upsilon_2$, $\mathbb{W}(\z,10;n)$ plots $\Upsilon_1$ as a much smaller magnitude than $\Upsilon_2$.

\begin{figure}[!ht]
\begin{center}
\includegraphics[width=0.325\textwidth]{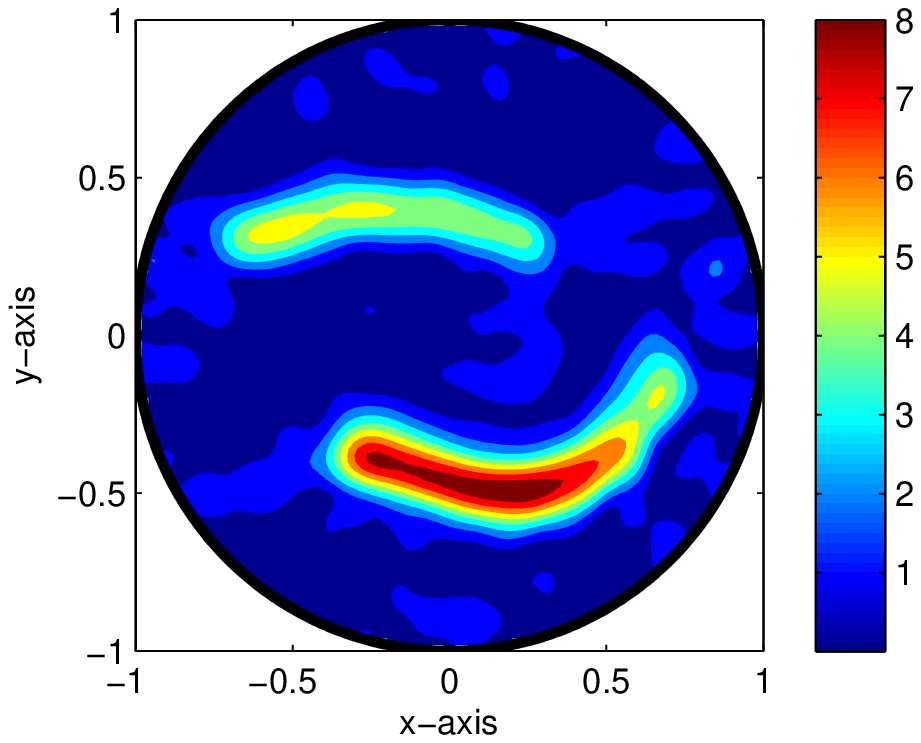}
\includegraphics[width=0.325\textwidth]{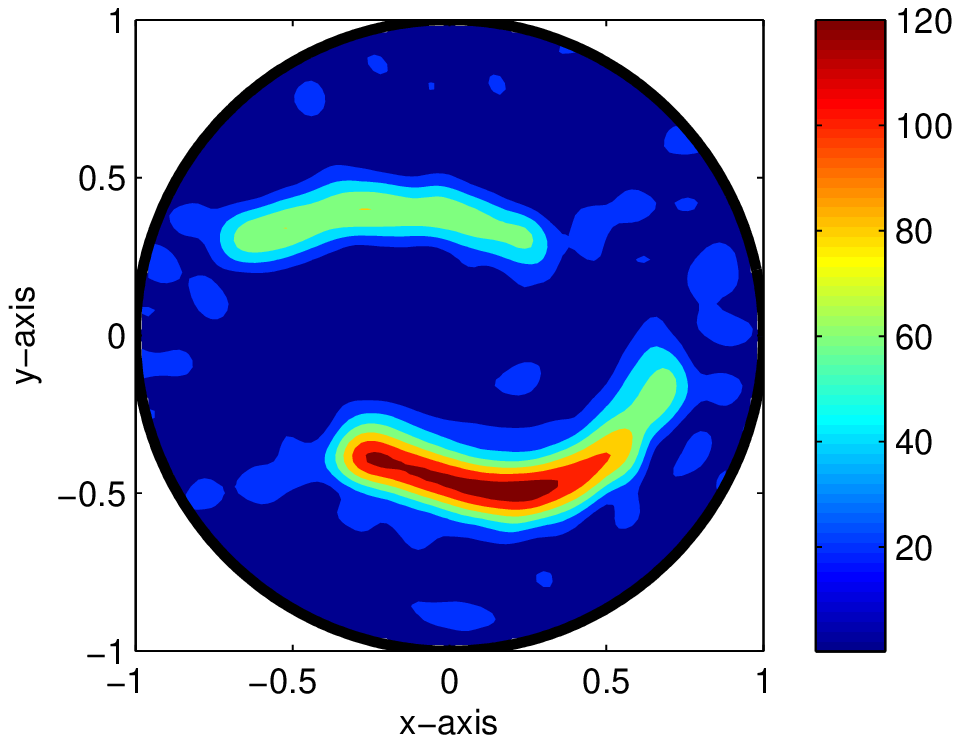}
\includegraphics[width=0.325\textwidth]{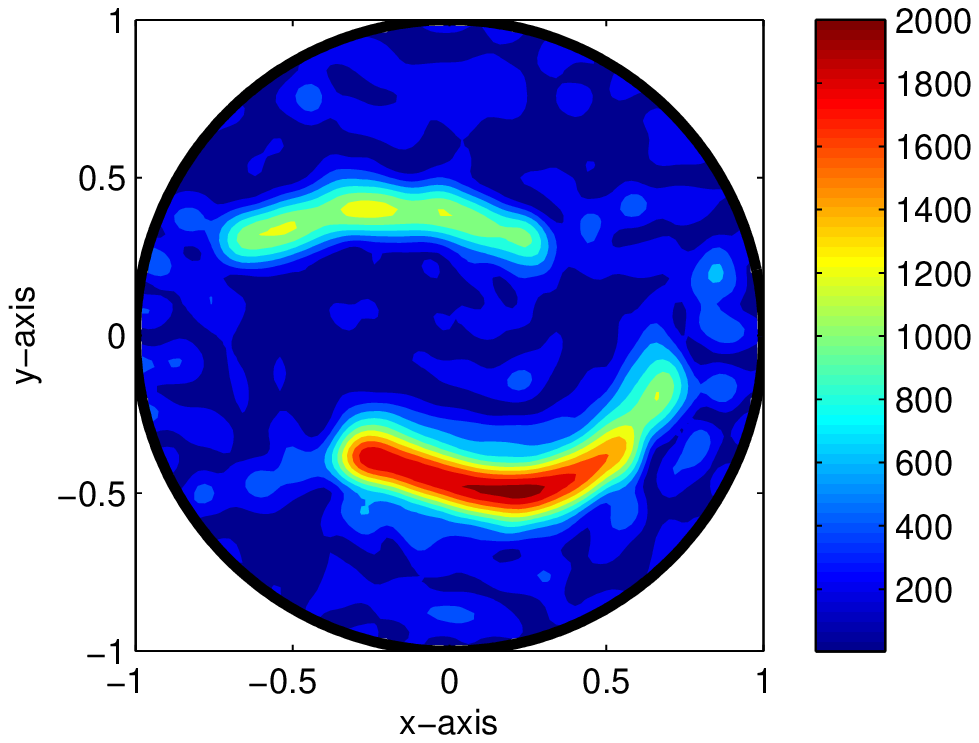}
\caption{\label{FigureTM}Same as Figure \ref{FigureT1} except that the thin inclusions are $\Upsilon_1\cup\Upsilon_2$.}
\end{center}
\end{figure}

\section{Conclusion and perspectives}\label{sec:5}
We considered a multi-frequency subspace migration algorithm weighted by applied frequencies for imaging thin electromagnetic inclusions and discussed its properties. The structure of subspace migration imaging functional and the corresponding numerical examples show that subspace migration imaging functional weighted by the power of applied frequencies is effective and robust against a large amount of random noise; however, application of a large number $n$ in (\ref{WeightedMultiFrequencyImagingFunction}) does not guarantee a good imaging result.

Although we considered a full-view inverse scattering problem, based on our contributions \cite{P2,PL2}, subspace migration imaging functional is also applicable to limited-view inverse problem. This fact has been verified mathematically in \cite{KP} for imaging of small targets but imaging of extended thin electromagnetic inclusions or perfectly conducting cracks has not been considered yet. Therefore, discovering certain properties of subspace migration imaging functional in limited-view problem will be an interesting research topic.

\section*{Acknowledgement}
The author would like to acknowledge Habib Ammari for many valuable advices as well as two-anonymous reviewers. This work was supported by the Basic Science Research Program through the National Research Foundation of Korea (NRF) funded by the Ministry of Education, Science and Technology (No. 2011-0007705), and the research program of Kookmin University in Korea.


\begin{thebibliography}{00}
\bibitem{ADIM}D. \`Alvarez, O. Dorn, N. Irishina, M. Moscoso, Crack reconstruction using a level-set strategy, J. Comput. Phys. 228 (2009), 5710--5721.
\bibitem{A}H. Ammari, An Introduction to Mathematics of Emerging Biomedical Imaging, in: Mathematics and Application Series, vol. 62, Springer-Verlag, Berlin, 2008.
\bibitem{AGJK}H. Ammari, J. Garnier, V. Jugnon, H. Kang, Stability and resolution analysis for a topological derivative based imaging functional, SIAM J. Control. Optim., 50 (2012), 48--76.
\bibitem{AGKPS}H. Ammari, J. Garnier, H. Kang, W.-K. Park, K. S{\o}lna, Imaging schemes for perfectly conducting cracks, SIAM J. Appl. Math, 71 (2011), 68--91.
\bibitem{AKLP}H. Ammari, H. Kang, H. Lee, W.-K. Park, Asymptotic imaging of perfectly conducting cracks, SIAM J. Sci. Comput., 32 (2010) 894--922.
\bibitem{BF}E. Beretta, E. Francini, Asymptotic formulas for perturbations of the electromagnetic fields in the presence of thin imperfections, Contemp. Math., 333 (2003), 49--63.
\bibitem{CHM}D. Colton, H. Haddar, P. Monk, The linear sampling method for solving the electromagnetic inverse scattering problem, SIAM J. Sci. Comput. 24 (2002), 719--731.
\bibitem{DL}O. Dorn, D. Lesselier, Level set methods for inverse scattering, Inverse Problems, 22 (2006), R67--R131.
\bibitem{GR}I. S. Gradshteyn, I. M. Ryzhik, Table of Integrals, Series, and Products, 2007, Academic Press.
\bibitem{G}R. Griesmaier, Multi-frequency orthogonality sampling for inverse obstacle scattering problems, Inverse Problems, 27 (2011), 085005.
\bibitem{JKHP}Y.-D. Joh, Y. M. Kwon, J. Y. Huh, W.-K. Park, Structure analysis of single- and multi-frequency subspace migrations in inverse scattering problems, Prog. Electromagn. Res., 136 (2013), 607--622.
\bibitem{JP}Y.-D. Joh, W.-K. Park, Structural behavior of the MUSIC-type algorithm for imaging perfectly conducting cracks, Prog. Electromagn. Res., 138 (2013), 211--226.
\bibitem{KR}A. Kirsch, S. Ritter, A linear sampling method for inverse scattering from an open arc, Inverse Problems, 16 (2000), 89--105.
\bibitem{KP}Y. M. Kwon, W.-K. Park, Analysis of subspace migration in the limited-view inverse scattering problems, Appl. Math. Lett., 26 (2013), 1107--1113.
\bibitem{MKP}Y.-K. Ma, P.-S. Kim, W.-K. Park, Analysis of topological derivative function for a fast electromagnetic imaging of perfectly conducing cracks, Prog. Electromagn. Res., 122 (2012), 311--325.
\bibitem{MP}Y.-K. Ma, W.-K. Park, A topological derivative based non-iterative electromagnetic imaging of perfectly conducting cracks, J. Electromagn. Eng. Sci., 12 (2012), 128--134.
\bibitem{P5}W.-K. Park, Improved subspace migration for imaging of perfectly conducting cracks, J. Electromagn. Waves Appl., in revision.
\bibitem{P4}W.-K. Park, Multi-frequency topological derivative for approximate shape acquisition of curve-like thin electromagnetic inhomogeneities, J. Math. Anal. Appl., 404 (2013), 501--518.
\bibitem{P1}W.-K. Park, Non-iterative imaging of thin electromagnetic inclusions from multi-frequency response matrix, Prog. Electromagn. Res., 106 (2010), 225--241.
\bibitem{P2}W.-K. Park, On the imaging of thin dielectric inclusions buried within a half-space, Inverse Problems, 26 (2010), 074008.
\bibitem{P3}W.-K. Park, Topological derivative strategy for one-step iteration imaging of arbitrary shaped thin, curve-like electromagnetic inclusions, J. Comput. Phys., 231 (2012), 1426--1439.
\bibitem{PL1}W.-K. Park, D. Lesselier, Electromagnetic MUSIC-type imaging of perfectly conducting, arc-like cracks at single frequency, J. Comput. Phys., 228 (2009), 8093--8111.
\bibitem{PL2}W.-K. Park, D. Lesselier, Fast electromagnetic imaging of thin inclusions in half-space affected by random scatterers, Waves Random Complex Media, 22 (2012), 3--23.
\bibitem{PL3}W.-K. Park, D. Lesselier, MUSIC-type imaging of a thin penetrable inclusion from its far-field multi-static response matrix, Inverse Problems, 25 (2009), 075002.
\bibitem{PL4}W.-K. Park, D. Lesselier, Reconstruction of thin electromagnetic inclusions by a level set method, Inverse Problems, 25 (2009), 085010.
\bibitem{PP}W.-K. Park, T. Park, Multi-frequency based direct location search of small electromagnetic inhomogeneities embedded in two-layered medium, Comput. Phys. Commun., 184 (2013), 1649--1659.
\bibitem{R}W. Rosenheinrich, Tables of Some Indefinite Integrals of Bessel Functions, \url{http://www.fh-jena.de/~rsh/Forschung/Stoer/besint.pdf}.
\end{thebibliography}
\end{document}